\documentclass[12pt,reqno]{amsart}
\usepackage{dsfont, amssymb,amsmath,amscd,latexsym, amsthm, amsxtra,amsfonts}
\usepackage[all]{xy}
\usepackage{graphicx}
\usepackage{float}
\usepackage{txfonts}
\usepackage[active]{srcltx}
\newtheorem{theorem}{Theorem}[section]
\newtheorem{lemma}[theorem]{Lemma}

\newtheorem{example}[theorem]{Example}

\newtheorem{Them}{Theorem}[section]

\begin{document}
\makeatletter
\def\@setauthors{%
\begingroup
\def\thanks{\protect\thanks@warning}%
\trivlist \centering\footnotesize \@topsep30\p@\relax
\advance\@topsep by -\baselineskip
\item\relax
\author@andify\authors
\def\\{\protect\linebreak}%
{\authors}%
\ifx\@empty\contribs \else ,\penalty-3 \space \@setcontribs
\@closetoccontribs \fi
\endtrivlist
\endgroup }
 \makeatother
 \baselineskip 17pt
\title[{{\small Optimal Control of a big financial company with Liability  }}]
 {{ Optimal control of a big financial company
  with debt liability under  bankrupt probability constraints}}
\author[{\bf  Z.Liang and B.Sun } ]
{ Zongxia Liang \\ Department of Mathematical Sciences, Tsinghua
University, Beijing 100084, China. Email:
zliang@math.tsinghua.edu.cn  \\
Bin Sun \\ Department of Mathematical Sciences, Tsinghua University,
Beijing 100084, China. Email: nksunbin@yahoo.com.cn}
\begin{abstract}
This paper considers an optimal control of a big financial company
with debt liability under bankrupt probability constraints. The
company, which faces constant liability payments and has choices to
choose various production/business policies from an available set of
control policies with different expected profits and risks, controls
the business policy and dividend payout process to maximize the
expected present value of the dividends until the time of
bankruptcy. However, if the dividend payout barrier is too low to be
acceptable, it may result in the company's bankruptcy soon. In order
to protect the shareholders' profits, the managements of the company
impose a reasonable and normal constraint on their dividend
strategy, that is, the bankrupt probability associated with the
optimal dividend payout barrier should be smaller than a given risk
level within a fixed time horizon. This paper aims at working out
the optimal control policy as well as optimal return function for
the company under  bankrupt probability constraint by stochastic
analysis, PDE methods and variational inequality approach. Moreover,
we establish a risk-based capital standard to ensure the capital
requirement of can cover the total given risk by numerical analysis
and give reasonable economic interpretation for the results.
 \vskip 10 pt
 \noindent {\bf MSC}(2000): Primary 91B30, 93E20, 65K10 ; Secondary
   60H05, 60H10.
 \vskip 10pt
 \noindent
 {\bf Keywords:} Regular-singular stochastic optimal control;
 Stochastic differential equations with reflection; Debt liability;
 Bankrupt probability constraints; Optimal dividend barrier;
 Dividend payout process.
 \end{abstract}
 \maketitle
\setcounter{equation}{0}
\section{{ {\bf Introduction}}}
 \vskip 10pt\noindent
In this paper, we study a model of a big financial company, which
has the possibility to choose various production/business policies
with different expected profits and associated risks. But at the
same time, the company has liability in which it has to pay out at a
constant rate no matter what the business plan is. The company
controls the business policy and dividend payout process to maximize
the expected present value of the dividends until the time of
bankruptcy. \vskip10pt \noindent Recently, there has been an upsurge
of interest in dealing with such optimal dividend control problems.
We refer readers to  Avanzi \cite{AV}(2009) and references therein,
H{\o}jgaard and Taksar\cite{s3, s4}(1999, 2001), Asmussen et
al.\cite{s7, s8}(1997, 2000), He and Liang\cite{ime02,
ime03}(2008,2009). For results in the model with debt liability, see
Choulli, Taksar and Zhou \cite{t1, t2, t3,s16, ime04}(2008, 2004,
2003, 2001, 2000). Guo, Liu and Zhou \cite{s1}(2004) is a
theoretical work on constrained nonlinear singular-regular
stochastic control problem. The optimal dividend payout for
diffusions with solvency constraints is solved in Paulsen
\cite{new23}(2003). According to Miller Modigliani, the managements
of the company choose the maximum of shareholders' return as their
goals. We see from these literatures that the optimality is achieved
by using an optimal dividend payout barrier $b$, i.e. whenever the
liquid reserve of the company goes above $b$, the excess is paid out
to the shareholders as dividends. However, the optimal dividend
barrier $b$ may be too low to be acceptable because this low
dividend payout barrier may result in the company's bankruptcy soon.
Thus, the company may be prohibited to pay dividends at such a low
level in order to avoid bankruptcy. So the managements of the
company impose some constraints on its dividends payout strategy.
One reasonable and normal constraint is that  the optimal dividend
barrier $b$  should be such that the bankrupt probability is not
larger than some predetermined risk level $\varepsilon$ within a
fixed time horizon $T$ and the cost for the safety is minimal.
\vskip10pt \noindent
 Based on the new idea, He, Hou and
Liang\cite{ime01}(2008) studied the linear regular-singular optimal
control problem of  insurance company with proportional reinsurance
policy under  bankrupt probability constraints as we state above.
They succeeded to find the optimal control policy under  bankrupt
probability constraints by proving the bankrupt probability is
decreasing in the dividend barrier and the existence of the dividend
barrier satisfying any given  bankrupt probability constraints.
Furthermore, Liang, Huang and Yao\cite{prep1,prep2, HL}(2010) gave a
exact result of such a control problem by proving the bankrupt
probability is continuous and strictly increasing w.r.t. the
dividend barrier $b$. These new results are mainly about the
insurance company with proportional reinsurance policy. Motivated by
these works, under any given bankrupt probability constraint, we are
interested in a common company, such as a big financial company,
insurance company, ...,  facing constant liability payments, which
controls the business policy and dividend payout process to maximize
the expected present value of the dividends until the time of
bankruptcy. Based on the relationship between those parameters that
govern the company's reserve process, we try to derive the optimal
control policy as well as optimal return function as well as a
risk-based capital standard to ensure the capital requirement of can
cover the total risk in several distinct cases of the qualitative
behavior of the company under some bankrupt probability constraints.
Moreover, we also give a robust analysis of the optimal return
function and the optimal dividend strategy w.r.t. the model
parameters and the constrained risk level of bankrupt probability.
 \vskip 10pt
\noindent The paper is organized as follows. In section 2, we
established a rigorous stochastic control  model of a big financial
company facing constant liability payments with some bankrupt
probability constraints. In section 3, we present main result of
this paper and its reasonable economic interpretation. In section 4,
we give risk analysis of the model we deal with to state the setting
treated in this paper is well defined and why we study such
regular-singular stochastic optimal control of the company. In
section 5, we give some numerical examples to portray how the debt
rate $\delta$, the constrained risk level of bankrupt probability,
the initial capital $x$ , the volatility rate $\sigma^2 $ and  the
profit rate $\mu$ impact on the optimal return function and the
optimal dividend strategy. We will list the properties of the
optimal return function and bankrupt probability in section 6. The
proof of main result will be given in section 7. The procedure of
solving the HJB equations and proofs of lemmas which are used to
prove the main result will be presented in the appendix.
 \vskip 10pt \noindent
\setcounter{equation}{0}
\section{{\small {\bf  Mathematical Model }}}
 \vskip 10pt\noindent
To give a mathematical formulation of our optimal control problem
treated in this paper, We start with a filtered probability space
$(\Omega, \mathcal {F}, \{ \mathcal {F}_{t}\}_{t\geq 0},
\mathbf{P})$ and a one-dimensional standard Brownian motion
$\{\mathcal {W}_t\}_{t\geq 0}$ on it, adapted to the filtration
$\mathcal{F}_t$. $\mathcal{F}_{t}$
 represents the information available at time $t$ and any decision
  made up to time $t$ is based on this information. For the
 { intuition} of our diffusion model we start from the classical
  Cram\'{e}r-Lundberg model of a reserve(risk) process. In this model
  claims arrive
according to a Poisson process $N_t$ with intensity $\lambda $ on
$(\Omega, \mathcal {F}, \{ \mathcal {F}_{t}\}_{t\geq 0},
\mathbf{P})$. The size of each claim is $X_i$. Random variables
$X_i$ are i.i.d. and are independent of the Poisson process $N_t$
with finite first and second moments given by $\mu$ and $ \sigma^2$
respectively. If there is no reinsurance and dividend payout, the
reserve (risk) process of insurance company is described by
$$ r_t= r_0 +pt-\sum^{N_t}_{i=1} X_i,$$
where $p$ is the premium rate. If $\eta > 0 $ denotes the {\sl
safety loading}, the $p$ can be calculated via the expected value
principle as
$$ p=(1+\eta ) \lambda \mu. $$
In a case where the insurance company shares risk with the
reinsurance, the sizes of the claims held by the insurer become
$X^{a}_i $, where $a$ is a (fixed) retention level.  For
proportional reinsurance, $a$ denotes the fraction of the claim
covered by the insurance company . Consider the case of {\sl cheap
reinsurance} for which the reinsuring company uses the same safety
loading as the insurance company, the reserve process of the
insurance company is given by
$$ r^{(a,\eta )}_t=u + p^{(a,\eta )}t - \sum^{N_t}_{i=1} X^{(a)}_i, $$
where
\begin{eqnarray*}
p^{(a,\eta)}=(1+\eta )\lambda \mathbb{E}\{X^{(a)}_i\}.
\end{eqnarray*}
Then by center limit theorem it is well known that for large enough
$\lambda $
\begin{eqnarray*}
\{r_{t}^{(a,\eta)}\}_{t\geq 0}\stackrel{d}{\approx} BM(\mu
at,\sigma^{2}a^2t).
\end{eqnarray*}
in $\mathcal{D}[0, \infty)   $ (the space of right continuous
functions with left limits endowed with the skorohod topology),
where $\mu=\eta\lambda E(X_i)$,  $\sigma =\sqrt{\lambda E(X_i^2)}$
and  $ BM ( \mu, \sigma^2)$ stands for Brownian motion with the
drift coefficient $\mu $ and diffusion coefficient $\sigma $ on
$(\Omega, \mathcal {F}, \{ \mathcal {F}_{t}\}_{t\geq 0},
\mathbb{P})$. So the passage to the limit works well in the presence
of a big portfolios, the  reserve (risk) process of the insurance
company can be approximated by the following diffusion process
\begin{eqnarray}\label{eq201} dR_{t}=\mu a(t)dt+\sigma
a(t)d\mathcal {W}_{t},
\end{eqnarray}
where  $a(t)$ denotes retention level. \noindent We refer readers
for this fact and for the specifies of the diffusion approximations
to Emanuel, Harrison and Taylor\\\cite{C1}(1975),
Grandell\cite{C2}(1977), Grandell\cite{C3}(1978),
Grandell\cite{C4}(1990), \\ Harrison\cite{C5}(1985),
Iglehart\cite{C6}(1969), Schmidli\cite{C7}(1994). \vskip
10pt\noindent In this paper, we consider a  common company which
faces constant liability payments. In view of the diffusion
approximations
 for the classical Cram\'{e}r-Lundberg model described above, we assume that
 the reserve
process of the company facing constant liability payments is given
by the following diffusion process
\begin{eqnarray}\label{202}
dR_t=(a(t)\mu-\delta)dt+a(t)\sigma d{\mathcal {W}}_t, \quad R_0=x
\end{eqnarray}
where $x$ is the initial reserve of the company, $\mu$ is the
expected profit per unit time (profit rate), $\sigma$ is the
volatility rate of the reserve process in the absence of any risk
control, $\delta$ represents the amount of money the company has to
pay per unit time (the debt rate) irrespective of what business
activities rate it chooses. In this model, the business activities
rate that the company chooses at time $t$ are modeled by $a(t)$,
which takes values in the interval $[\alpha, \beta]$, where
 $0<\alpha<\beta<+\infty$. The restriction $\alpha>0$ reflects the
fact that there are statutory reasons that its business activities
rate cannot be reduced to zero, unless the company faces bankruptcy.
 \vskip 10pt \noindent
In our model, a policy $\pi$ is a pair of non-negative
c\`{a}dl\`{a}g $ \mathcal {F}_{t}$-adapted  processes $\{a_\pi
(t),L_t ^\pi\}$, where $a_\pi (t)\in [\alpha, \beta]$ corresponds to
the risk exposure at time $t$ and $L_t ^ \pi$ corresponds to the
cumulative amount of dividend pay-outs distributed up to time $t$. A
policy $\pi =\{a_\pi (t),L_t ^\pi\} $ is called admissible if
$\alpha \leq a_\pi (t)\leq \beta $ and $L_t ^\pi$ is a nonnegative,
non-decreasing, right-continuous function. We denote $\Pi$ the set
of all admissible policies. When a admissible policy $\pi
 $ is applied, the resulting reserve process is denoted by $\{ R_t^\pi
 \}$. In view of (\ref{202}) we rewrite
  $R_t^\pi$ as follows
\begin{eqnarray}\label{203}
dR^\pi_t=(a_\pi(t)\mu-\delta)dt+a_\pi(t)\sigma d{\mathcal
{W}}_t-dL^\pi_t, \quad R^\pi_0=x,
\end{eqnarray}
where $x(>0)$ is the initial (capital) reserve. In addition, we
assume the company needs to keep its reserve above 0 to avoid
bankruptcy. For the given control policy $\pi$, we define the time
of bankruptcy as $\tau ^\pi _x=\inf\big \{t\geq0: R^\pi_{t}\leq
0\big \}$. $\tau ^\pi _x$ is an $\mathcal {F}_{t}$ -stopping time.
 \vskip 10pt \noindent
The objective of the company is to maximize the expected present
value of the dividends payout until the time of bankrupt by choosing
control policy $\pi$ from the admissible set $\Pi$. Choulli, Taksar
and Zhou\cite{t1}(2003) proved that there exists an optimal dividend
barrier $b_0$ and the optimal policy $\pi^*_{b_0}=
\{a_{\pi^*_{b_0}}(\cdot),L^{\pi^*_{b_0}}_\cdot\}$, which maximize
the expected present value of the dividends payout before
bankruptcy, i.e., $b_0$ is such that
 \begin{eqnarray}\label{204}
 J(x,\pi)=\mathbf{E}\big \{\int_0^{\tau^\pi_x } e^{-ct} dL_t^\pi\big\},
 \end{eqnarray}
 \begin{eqnarray}\label{205}
 V(x,b_0)=\sup_{\pi \in \Pi}\{J(x,\pi)\}=J(x,\pi^*_{b_0}),
 \end{eqnarray}
where $c$ denotes the discount rate. If the optimal dividend barrier
is too low, the bankrupt probability within a fixed time will be
bigger than a given level. This is not acceptable by the management
of the company. Taking the balance of profit and risk into
consideration, we impose small bankrupt probability constraint on
the company's control policy. We describe our optimal control
problem as follows.
\vskip 5pt\noindent
 Let $\Pi_b=\big\{\pi\in \Pi :
\int _0 ^{ \infty }{I}_{\{s:R^\pi(s)<b\}}dL_s^{\pi}=0\big\} $ for
$b\geq 0$ . Then it is easy to see that $\Pi=\Pi_0$ and
$b_1>b_2\Rightarrow \Pi_{b_1}\subset \Pi_{b_2}$. For a given
admissible policy $\pi$, we define the optimal return function
$V(x)$ by
\begin{eqnarray}\label{eq206}
J(x,\pi)&=&{\bf E}\big \{\int_0^{\tau^\pi_x } e^{-ct} dL_t^\pi\big\},\nonumber \\
 V(x,b )&=&\sup_{\pi \in \Pi_b}\{J(x,\pi)\},
\end{eqnarray}
\begin{eqnarray}\label{eq207}
 V(x)&=&\sup_{b \in \mathfrak{B}}\{  V(x,b )  \}
\end{eqnarray}
and the optimal policy $\pi^* $ by
\begin{eqnarray}\label{eq208}
J(x,\pi^*)= V(x),
\end{eqnarray}
 where
  \begin{eqnarray*}
 \mathfrak{B}:=\big \{b\ :\ \mathbb{P}[\tau_{b}^{\pi_{b}} \leq T] \leq
\varepsilon \ , \  J(x, \pi_b)= V(x,b) \mbox{ and} \ \pi_b \in
\Pi_b\big\},
  \end{eqnarray*}
 $c>0$ is a discount rate, $\tau_b^{\pi_b}$ is
 the time of bankruptcy $\tau_x^{\pi_b} $  when
the initial reserve $x=b$ and the control policy is $\pi_b$.
$\varepsilon$ is a given  constrained risk  level of bankrupt
probability. $1-\varepsilon$ is the standard of security and less
than solvency for a given risk level $\varepsilon>0$. $ \mathfrak{B}
$ is called the risk constrained set. \vskip 10pt\noindent The main
purpose of this paper is to derive the optimal return function
$V(x)$, the optimal policy $\pi^*$ as well as the optimal dividend
payout barrier $b^*$ and try to give a reasonable economic
interpretation and discuss effect of the debt rate $\delta$, the
constrained risk level $\varepsilon $ of bankrupt probability, the
initial capital $x$, the volatility rate $\sigma^2 $ and  the profit
rate $\mu$ on the optimal return function and the optimal dividend
strategy $\pi^*$.
 \vskip 5pt \noindent
 \setcounter{equation}{0}
\section{{{\bf Main Result}} }
 \vskip 5pt \noindent
 In this section we first present main result of this paper, then,
 together with numerical results in section 5 below,
  give its reasonable economic interpretation.
 The proof of the main result
  will be given in section 7.
\begin{theorem}\label{theorem31}
Let $\varepsilon\in (0,1)$ be the constrained risk level of bankrupt
probability and time horizon $T$ be given. \vskip 10pt\noindent(i)
If
 $\mathbf{P}[\tau_{b_0}^{\pi^*_{b_0}}\leq T]\leq \varepsilon$, then the
optimal return function $V(x)$ is $f(b_0, x)$ defined by (\ref{61})
below, and $V(x)=f(b_0, x)=J(x,\pi_{b_o}^\ast) $. The optimal policy
$\pi_{b_o}^\ast $ is
$\{a^\ast_{b_o}(R^{\pi_{b_o}^\ast}_t),L^{\pi_{b_o}^\ast}_t\}$, where
$\{R^{\pi_{b_o}^\ast}_t, L^{\pi_{b_o}^\ast}_t  \} $ is uniquely
determined  by the following stochastic differential
 equation
\begin{eqnarray}\label{301}
\left\{
\begin{array}{l l l}
dR_t^{\pi_{b_o}^\ast}=(a^*_{b_o}(R^{\pi_{b_o}^\ast}_t)\mu-\delta)dt+
\sigma a^*_{b_o}(R^{\pi_{b_o}^\ast}_t)d{\mathcal {W}}_{t}-dL_t^{\pi_{b_o}^\ast},\\
R_0^{\pi_{b_o}^\ast}=x,\\
0\leq R^{\pi_{b_o}^\ast}_t\leq  b_0,\\
\int^{\infty}_0 I_{\{t: R^{\pi_{b_o}^\ast}_t
<b_0\}}(t)dL_t^{\pi_{b_o}^\ast}=0.
\end{array}\right.
\end{eqnarray}
The solvency of the company is bigger than $1-\varepsilon$.
 \vskip 10pt\noindent
(ii) If $\mathbf{P}[\tau_{b_0}^{\pi^*_{b_0}}\leq T]>\varepsilon $,
there is a unique optimal dividend $b^\ast(\geq b_0)$ satisfying
$\mathbf{P}[\tau _{b^\ast}^{\pi_{b^*}^\ast}\leq T]= \varepsilon $.
The optimal return  function $V(x)$ is $g(x,b^*)$ defined by
(\ref{64}), that is,
\begin{eqnarray}\label{eq302}
V(x)= g(x,b^*) =\sup_{b\in \mathfrak{B}}\{V(x,b)\}
\end{eqnarray}
and
\begin{eqnarray}\label{eq303}
b^* \in \mathfrak{B}:=\big \{b: \mathbb{P}[\tau_{b}^{\pi^*_{b}} \leq
T] \leq \varepsilon, \  J(x, \pi^*_b)= V(x,b) \mbox{ and} \ \pi^*_b
\in \Pi_b\ \big\}.
\end{eqnarray}
The optimal policy $\pi_{b^*}^\ast=
 \{a^\ast_{b^*}(R^{\pi_{b^*}^\ast}_t),L^{\pi_{b^*}^\ast}_t\}$, where
$\{R^{\pi_{b^*}^\ast}_t, L^{\pi_{b^*}^\ast}_t  \} $ is uniquely
determined  by the following stochastic differential equation
\begin{eqnarray}\label{eq304}
\left\{
\begin{array}{l l l}
dR_t^{\pi_{b^*}^\ast}=(\mu a^*_{b^*}(R^{\pi_{b^*}^\ast}_t)-\delta
)dt+\sigma a^*_{b^*}(R^{\pi_{b^*}^\ast}_t)d
{W}_{t}-dL_t^{\pi_{b^*}^\ast},\\
R_0^{\pi_{b^*}^\ast}=x,\\
0\leq R^{\pi_{b^*}^\ast}_t\leq  b^*,\\
\int^{\infty}_0 I_{\{t: R^{\pi_{b^*}^\ast}_t
<b^*\}}(t)dL_t^{\pi_{b^*}^\ast}=0.
\end{array}\right.
\end{eqnarray}
The solvency of the company is $1-\varepsilon$. \vskip 10pt\noindent
(iii) Moreover,
\begin{eqnarray}\label{eq305}
\frac{g(x,b^*)}{g(x,b_0)}\leq 1,
 \end{eqnarray}
 where $a^*_b(x)$ is defined by (\ref{66}) below.
 \end{theorem}
\vskip 10pt\noindent
 { \bf Economic interpretation of Theorem
\ref{theorem31} is as follows.}
 \vskip 5pt\noindent {\sl (1) For a
given  constrained risk level $  \varepsilon $  of bankrupt
probability
  and time horizon $T$, if the company's bankruptcy
probability
  is less than this given risk constraint level $\varepsilon$, the optimal
  control problem of
(\ref{eq206}) and (\ref{eq207}) is the traditional problem
(\ref{204}) and (\ref{205}). The  bankrupt probability constraints
here do not work. \vskip 10pt\noindent
 (2) If the company's bankruptcy probability
  is larger than this given constrained risk level $\varepsilon$,
    the traditional optimal
control policy fails to meet the requirement of  bankrupt
probability constraint. So the company has to find an optimal policy
$\pi_{b^*}^\ast $ to improve its solvency ability and ensure the
company  operates safely. The optimal reserve process $
R^{\pi_{b^*}^\ast}_t $ is a diffusion process reflected at a
dividend barrier $b^*$, and the process $L^{\pi_{b^*}^\ast}_t $ is
the dividend payout process that ensures the reflection. $a^*_{b^*}$
is the optimal feedback control function. The optimal policy means
that the company will pay out the amount of reserve in excess of
$b^*$ as dividend once its reserve is bigger than $b^*$. Under this
control policy, we can guarantee that the company's bankrupt
probability can stay below $\varepsilon$.
 \vskip 5pt\noindent
(3)The inequality (\ref{eq305}) shows that the optimal control
policy $\pi_{b^*}^\ast $ will decrease the company's profit-earning
capability. We can treat this decrease of the profit as the cost of
keeping the company's risk at a low level and the cost,
$g(x,b_0)-g(x, b^*)$, is minimal in view of ( \ref{eq303}), Lemma
\ref{lemma63} and Lemma \ref{64} below. Thus $\pi_{b^*}^\ast $ is
the best equilibrium control policy between making profit and
improving solvency. \vskip 10pt\noindent (4) From  the figure
\ref{g_de} in Example \ref{ex51} below we  see that the  optimal
return function $g(x)$ is decreasing w.r.t. the debt rate $\delta$.
Figure \ref{g_de} shows that the higher debt rate will lessen the
company's profit, so the company should keep its debt rate at a
appropriate level.\vskip 5pt\noindent
 (5)We can see from figure \ref{b_ep} in Example \ref{ex52} below that
the optimal dividend barrier $b^*$ is decreasing w.r.t. the
constrained  risk  level $  \varepsilon $  of bankrupt probability.
And the optimal dividend barrier $b^*$ is uniquely decided by
$\mathbf{P}[\tau _{b^\ast}^{\pi_{b^*}^\ast}\leq T]= \varepsilon $,
i.e., $1-\phi^{b^*}(b^*,b^*)=\varepsilon$ (see Lemma \ref{lemma65}
below). The the optimal dividend barrier $b^*$ is also increasing
function of the volatility $\sigma^2$ (see the figure \ref{b_si} in
Example \ref{ex53} below ). \vskip 5pt\noindent (6) We call
$R_0:=x_{b^*}(\varepsilon)$ the risk-based capital standard to
ensure the capital requirement of can cover the total given risk
$\varepsilon $, where $x_{b^*}(\varepsilon)$ is determined by
$1-\phi^{b^*}(x,b^*)=\varepsilon$ (see Lemma \ref{lemma65}).
 We see from the figure \ref{x_ep} in Example \ref{ex54} below that
 risk-based capital $x_{b^*}(\varepsilon)$ decreases with
 risk level $\varepsilon $. Since the optimal feedback control function
 $a^*_{b^*}(x)$ is increasing w.r.t. $x$, in view of comparison theorem
 for SDE, the constrained risk level $  \varepsilon $  lessens
the optimal business activities rate $a^*(t)$, but improves dividend
payout
 process $L^*_t(\varepsilon)$.
 It also lessens the optimal return function(see Example
\ref{ex57} below).
 \vskip
5pt\noindent (7) We can see from the figures \ref{g_mu} and
\ref{g_si} below that the optimal return function $g(x)$ is
increasing in both the profit rate $\mu$ and the volatility rate
$\sigma$. }
 \vskip 10pt \noindent
 \setcounter{equation}{0}
\section{{ {\bf Risk Analysis}} }
 \vskip 8pt \noindent
In this section, we proceed a risk analysis on the model we are
studying. We first work out the lower boundary of bankrupt
probability when we applied $b_0$ as the dividend barrier. It proves
that the risk constrained set $\mathfrak{B}$ is not $\Re_+=[0,
+\infty )$. So the topic of this paper is fundamental to studying
the optimal control problem under bankrupt probability constraints.
It also states that the company has to find optimal policy to
improve its solvency.
\begin{Them}\label{theorem41}
Let
 $\{R^{\pi_{b_0}^\ast}_t, L^{\pi_{b_0}^\ast}_t  \} $ be the unique solution of the
following SDE( see Lions and Sznitman \cite{s10100}(1984))
\begin{eqnarray}\label{eq401-theorem41}
\left\{
\begin{array}{l l l}
dR_t^{\pi_{b_o}^\ast}=(\mu a^*_{b_o}(R^{\pi_{b_o}^\ast}_t)-\delta
)dt+\sigma a^*_{b_o}(R^{\pi_{b_o}^\ast}_t)d
{W}_{t}-dL_t^{\pi_{b_o}^\ast},\\
R_0^{\pi_{b_o}^\ast}=b_0,\\
0\leq R^{\pi_{b_o}^\ast}_t\leq  b_0,\\
\int^{\infty}_0 I_{\{t: R^{\pi_{b_o}^\ast}_t
<b_0\}}(t)dL_t^{\pi_{b_o}^\ast}=0.
\end{array}\right.
\end{eqnarray}
Then
\begin{eqnarray}\label{eq402-theorem41} {\bf P}(\tau_{b_0}^{\pi^*_{b_0}}\leq
T)&\geq&
\frac{4[1-\Phi(\frac{b_0}{\alpha\sigma\sqrt{T}})]^2}{\exp\{
\frac{T}{\sigma^2}{\max\{\mu-\frac{\delta}{\beta},|\mu-\frac{\delta}{\alpha}|\}}^2\}}
 \nonumber\\&\equiv& \varepsilon_0(b_0, \mu, \delta, \sigma, T, \alpha, \beta)
>0,
\end{eqnarray}
where $\tau ^{\pi^*_{b_o}} _{b_o}=\inf\big \{t\geq0:
R^{\pi^{*}_{b_o}}_{t}\leq 0\big \}$ and $\Phi(\cdot)$ is the
standard normal distribution function.
\end{Them}
\begin{proof}
Since $a^*_{b_o}(x)$ (defined by (\ref{66})) is a bounded Lipschitz
continuous function w.r.t. $x$, the following SDE
\begin{eqnarray}\label{eq403-theorem41}
\left\{
\begin{array}{l l l}
dX_t=(\mu a^*_{b_0}(X_t)-\delta )dt+\sigma a^*_{b_0}(X_t)d
{W}_{t},\\
X_0=b_0\\
\end{array}\right.
\end{eqnarray}
has a unique solution $X_t$. By comparison theorem for
one-dimensional It\^{o} process, we have
\begin{eqnarray}\label{eq404-theorem41}
{\bf{P}}[R^{\pi_{b_o}^\ast}_t \leq X_t]=1
\end{eqnarray}
Define a measure ${\bf Q}$ on $\mathcal{ F}_T  $ by
\begin{eqnarray}\label{eq405-theorem41}
 d{\bf{Q}}(\omega)=M(T)d{\bf P}(\omega )\label{eq405-theorem41}
\end{eqnarray} where
\begin{eqnarray*}
M(t)&\equiv&\exp\big\{-\int_{0}^{t}\frac{(\mu
a^*_{b_0}(X_t)-\delta)}{\sigma
a^*_{b_0}(X_t)}dW_{s}\\
&-& \frac{1}{2}\int_{0}^{t}\frac{(\mu
a^*_{b_0}(X_t)-\delta)^{2}}{[\sigma a^*_{b_0}(X_t)]^{2}}ds\big \}.
\end{eqnarray*}
Since $\{ M(t)\}$ is a martingale w.r.t. $\mathcal {F}_t$, we have
${\bf E} \big [ M(T) \big ] =1$. Moreover, noticing that
$a^*_{b_o}(X_t)$ belongs to $[\alpha, \beta]$, we obtain
\begin{eqnarray}\label{eq406-theorem41}
{\bf{E}}^{\bf{P}}[M(T)^{2}]\leq \exp\{
\frac{T}{\sigma^2}{\max\{\mu-\frac{\delta}{\beta},|\mu-\frac{\delta}{\alpha}|\}}^2\}
\end{eqnarray}
Using Girsanov theorem,  ${\bf Q}$ is a probability measure on
$\mathcal {F}_T$ and the process $\{X_t\}$ satisfies the following
SDE
\begin{eqnarray}\label{eq407-theorem41}
dX_t=\sigma a^*_{b_0}(X_t)d\tilde{W}_t,\  X_0=b_0
\end{eqnarray}
where $\tilde{W}_t = W_t+\int_{0}^{t}\frac{(\mu
a^*_{b_0}(X_s)-\delta)}{\sigma a^*_{b_0}(X_s)}d{s}$. It is easy to
see that $\tilde{W}_t$
 is a Brownian motion  on $(\Omega, \mathcal {F}, \{ \mathcal {F}_{t}\}_{t\geq 0},
{\bf Q})$.
\vskip 5pt\noindent Define a time changes $\rho(t)$ by
\begin{eqnarray}\label{eq408-theorem41}
\dot{\rho}(t)=\frac{1}{{a^*_{b_0}}^2(X_t)\sigma^2},
\end{eqnarray}
 then $\rho(t)$ is a
strictly increasing function. If we denote $X_{\rho(t)}$ by
$\hat{X}_t$, then we have
\begin{eqnarray*}
\hat{X}_t=b_0+\hat{W}_t.
\end{eqnarray*}
 Noticing
that $0<\alpha \leq {a^*_{b_0}}(R_t) \leq \beta < +\infty$, we get
\begin{eqnarray}\label{eq409-theorem41}
\frac{1}{\beta^2\sigma^2}\leq \dot{\rho}(t)\leq
\frac{1}{\alpha^2\sigma^2}.
\end{eqnarray}
Due to the fact $\rho(t)=\int_{0}^{t}\dot{\rho(s)}ds$, we can deduce
that $\rho(t)\leq \frac{t}{\alpha^2\sigma^2}$ and $\rho^{-1}(t) \geq
\alpha^2\sigma^2 t$, where $\rho^{-1}$ denotes the inverse of
$\rho$. Then we have
\begin{eqnarray}\label{eq410-theorem41}
{\bf{Q}}[\inf\{t: X_t\leq 0\}\leq T]&=&{\bf{Q}}[\inf\{t:
\hat{X}_{\rho^{-1}(t)}\leq 0\}\leq
T]\nonumber\\&=&{\bf{Q}}[\inf\{\rho(t): b_0+\hat{W}_t\leq 0\}\leq T]\nonumber\\
&=&{\bf{Q}}[\inf \{t: \hat{W}_t\leq - b_0 \}\leq
\rho^{-1}(T)]\nonumber\\&\geq& {\bf{Q}}[\inf\{t: \hat{W}_t\leq
-b_0\}\leq \alpha^2\sigma^2
T]\nonumber\\&=&2[1-\Phi(\frac{b_0}{\alpha \sigma\sqrt{T}})]>0.
\end{eqnarray}
Using H\"{o}lder inequalities as well as (\ref{eq405-theorem41}),
\begin{eqnarray}\label{eq411-theorem41}
{\bf{Q}}[\inf\{t: X_t\leq 0\}\leq T]&=&
\int_{\Omega}{\bf{1}}_{[\inf\{t: X_t\leq 0\}\leq T]}d\bf{Q}(\omega)
\nonumber\\&=&\int_{\Omega}{\bf{1}}_{[\inf\{t: X_t\leq 0\}\leq T]}M_Td{\bf{P}}(\omega)\nonumber\\
&=& {\bf{E}}^{{\bf{P}}}[M_T {\bf{1}}_{[\inf\{t: X_t\leq 0\}\leq T]}]
\\&\leq &
{\bf{E}}^{{\bf{P}}}[M^2_T]^{\frac{1}{2}}{\bf{P}}[\inf\{t: X_t\leq
0\}\leq T]^{\frac{1}{2}}.\nonumber
\end{eqnarray}
Substituting (\ref{eq406-theorem41}) and (\ref{eq410-theorem41})
into (\ref{eq411-theorem41}), we get
\begin{eqnarray*}
{\bf{P}}[\inf\{t: X_t\leq 0\}\leq T] &\geq & \frac{{\bf{Q}}[\inf\{t:
X_t\leq 0\}\leq T]^2}{{\bf{E}}^{\bf{P}}[M_T^{2}]} \\
&\geq& \frac{4[1-\Phi(\frac{b_0}{\alpha\sigma\sqrt{T}})]^2}{\exp\{
\frac{T}{\sigma^2}{\max\{\mu-\frac{\delta}{\beta},|\mu-\frac{\delta}{\alpha}|\}}^2\}}.
\end{eqnarray*}
By virtue of (\ref{eq404-theorem41}), we have
\begin{eqnarray}\label{eq412-theorem41}
{\bf P}[\tau_{b_0}^{\pi^*_{b_0}}\leq T] &=& {\bf{P}}[\inf \{t:
R^{\pi^{*}_{b_o}}_{t}\leq 0\}\leq T] \nonumber\\ &\geq&
{\bf{P}}[\inf\{t: X_t\leq 0\}\leq T]\nonumber\\ &\geq&
\frac{4[1-\Phi(\frac{b_0}{\alpha\sigma\sqrt{T}})]^2}{\exp\{
\frac{T}{\sigma^2}{\max\{\mu-\frac{\delta}{\beta},|\mu-\frac{\delta}{\alpha}|\}}^2\}}\nonumber\\
&\equiv& \varepsilon_0(b_0, \mu, \delta, \sigma, T, \alpha, \beta)\nonumber\\
&>& 0.
\end{eqnarray}
\end{proof}

\noindent {\bf  The economic interpretation of  theorem
 \ref{theorem41} is the following.}\vskip 10pt \noindent
 Assume $\frac{2\delta}{\mu}<\alpha$, then we have
 $\mu-\frac{\delta}{\alpha}>0$. So the lower  boundary of the
 bankrupt probability $\varepsilon_0(b_0, \mu, \delta, \sigma, T, \alpha,
 \beta)$ becomes to $\frac{4[1-\Phi(\frac{b_0}{\alpha\sigma\sqrt{T}})]^2}{\exp\{
\frac{T}{\sigma^2}(\mu-\frac{\delta}{\beta})^2\}}$. Based on the
assumption, we have the following explanations.\vskip 10pt \noindent
 {\sl   (1)\ The lower  boundary of bankrupt probability for the company
 $\varepsilon_0(b_0, \mu, \delta, \sigma,\\ T, \alpha, \beta)$
is an increasing function of $(\sigma, \delta, \alpha)$, which means
 that higher volatility rate $\sigma$ and debt rate $\delta$ will make the company
 face larger risk. In addition, risk will increase as the lower
 boundary $\alpha$ of control function $a(x)$ increases.
 \vskip 10pt
\noindent(2) \ The lower  boundary of bankrupt probability for the
company
 $\varepsilon_0(b_0, \mu, \delta, \sigma,\\ T, \alpha, \beta)$ is decreasing
in $ ( b_0, \mu, \beta )$, which means that paying dividends at a
lower barrier will cause larger bankrupt probability. On the other
hand, the higher the profit rate is, the lower the risk is.
Improving the upper boundary $\beta $ of the control function $a(x)$
can also reduce the company's risk. \vskip 10pt \noindent (3) The
company has a positive bankrupt probability within the time interval
$[0,T]$ if we set $b_0$ as the dividends barrier. In order to keep
the company's risk at a low level, we need adjust our control policy
and find the optimal dividends barrier $b^*$ under lower constrained
risk level of  bankrupt probability. } \vskip 10pt\noindent The
second result is the following, which states that the risk
constrained set $\mathfrak{B} $ defined in section 2 is non-empty
for any $ \varepsilon >0$, together with the first result,  also
guarantees our problem (\ref{eq206}), (\ref{eq207}), (\ref{eq208})
is well defined.
\begin{Them}\label{theorem42}
Let $( R^{\pi_{b}^\ast}_t, L_t^{\pi_{b}^\ast}    )$ be defined by
\begin{eqnarray}\label{eq401-theorem42}
\left\{
\begin{array}{l l l}
dR_t^{\pi_{b}^\ast}=(\mu a^*_{b}(R^{\pi_{b}^\ast}_t)-\delta
)dt+\sigma a^*_{b}(R^{\pi_{b}^\ast}_t)d
{W}_{t}-dL_t^{\pi_{b}^\ast},\\
R_0^{\pi_{b}^\ast}=b,\\
0\leq R^{\pi_{b}^\ast}_t\leq  b,\\
\int^{\infty}_0 I_{\{t: R^{\pi_{b}^\ast}_t
<b\}}(t)dL_t^{\pi_{b}^\ast}=0,
\end{array}\right.
\end{eqnarray}
and $\tau_{b}^{\pi_b^*}=\inf\{t\geq 0: R^{\pi_{b}^\ast}_t<0 \}$.
Then
\begin{eqnarray}\label{eq402-theorem42}
\lim_{b\rightarrow  \infty}{\bf P}[ \tau_{b}^{b}\leq T]=0.
\end{eqnarray}
\end{Them}
\begin{proof}

For any $b\geq 1$, we have $b\geq \sqrt b$. By comparison theorem on
SDE (see Ikeda and Watanabe \cite{IW}(1981)), we have
\begin{eqnarray}\label{eq403-theorem42}
  \mathbf{P}[\tau_{b}^{\pi^*_{b}}\leq T]\leq \mathbf{P}[\tau_{\sqrt b}^{\pi^*_{b}}\leq T].
\end{eqnarray}
Let $R^{(1)}_t$ satisfy the following SDE,
\begin{eqnarray}\label{eq404-theorem42}
\left\{
\begin{array}{l l l}
  dR_t^{(1)}=(a^*(R_t^{(1)})\mu -\delta)dt + a^*(R_t^{(1)})\sigma d\mathcal {W}_t,\\
  R_0^{(1)}=\sqrt b.
  \end{array}\right.
\end{eqnarray}
Then, we have
\begin{eqnarray}\label{eq405-theorem42}
 \mathbf{P}[\tau_{\sqrt b}^{\pi^*_b}\leq T]
  &\leq&\mathbf{P}[R_t^{(1)}=0 \ \mbox{or}\ R_t^{(1)}
  =b\ \mbox{for some $0\leq t\leq T$ }]\nonumber\\
  &\leq& \mathbf{P}[\sup_{0\leq t \leq T} R_t^{(1)}\geq
  b]+\mathbf{P}[\inf_{0\leq t \leq T} R_t^{(1)}\leq 0].
\end{eqnarray}
Next, we estimate $\mathbf{P}[\sup_{0\leq t \leq T} R_t^{(1)}\geq
  b]$ and $\mathbf{P}[\inf_{0\leq t \leq T} R_t^{(1)}\leq 0]$, respectively.
H\"{o}lder's inequality and $a^*(x)\leq \beta$ yield that
\begin{eqnarray}\label{eq406-theorem42}
   \sup_{0\leq t \leq T} (R_t^{(1)})^2 &\leq&
   3 (\sqrt b)^2 +3 \sup_{0\leq t \leq T}(\int_0^t (a^*(R_s^{(1)})\mu -\delta)ds)^2+
   \nonumber\\
   &&3 \sup_{0\leq t \leq T}(\int_0^t a^*(R_t^{(1)})\sigma d\mathcal {W}_s)^2\nonumber \\
   &\leq& 3b+ 3(\beta \mu-\delta)^2T^2+ \nonumber\\
   &&3 \sup_{0\leq t \leq T}(\int_0^t a^*(R_t^{(1)})\sigma d\mathcal {W}_s)^2.
\end{eqnarray}
By Markov inequality, B-D-G inequalities and
(\ref{eq406-theorem42}), we obtain
\begin{eqnarray}\label{eq407-theorem42}
   \mathbf{P}[\sup_{0\leq t \leq T} R_t^{(1)}\geq b] &\leq&
  \frac{\mathbf{E}[\sup_{0\leq t \leq T} (R_t^{(1)})^2]}{b^2}\nonumber \\
  &\leq& \frac{3b+3(\beta\mu-\delta)^2T^2+12{\beta}^2{\sigma}^2T}{b^2}\nonumber\\
  & \rightarrow &0, \ \ \mbox{as}\  b \rightarrow  \infty.
\end{eqnarray}
Now we turn to estimating $\mathbf{P}[\inf_{0\leq t \leq T}
R_t^{(1)}\leq 0]$. Let $R_t^{(2)}$ satisfy the following SDE
\begin{eqnarray}\label{eq408-theorem42}
\left\{
\begin{array}{l l l}
  dR_t^{(2)}=(a^*(R_t^{(1)})\mu -\delta)dt + a^*(R_t^{(1)})\sigma d\mathcal {W}_t,\\
  R_0^{(2)}=0.
  \end{array}\right.
\end{eqnarray}
Thus we have
\begin{eqnarray}\label{eq409-theorem42}
  R_t^{(1)} = \sqrt b +R_t^{(2)}.
\end{eqnarray}
Therefore, by using the same argument as in (\ref{eq407-theorem42})
, we get
\begin{eqnarray}\label{eq410-theorem42}
   \mathbf{P}[\inf_{0\leq t \leq T} R_t^{(1)}\leq 0] &=&
    \mathbf{P}[\inf_{0\leq t \leq T} R_t^{(2)}\leq -\sqrt b]\nonumber\\
    &=&\mathbf{P}[\sup_{0\leq t \leq T}(-R_t^{(2)})\geq \sqrt b]\nonumber\\
    &\leq &\frac{\mathbf{E}[\sup_{0\leq t \leq T}(-R_t^{(2)})^2]}{(\sqrt b)^2}
    \nonumber\\
    &\leq& \frac{2(\beta \mu-\delta)^2T^2+8{\beta}^2{\sigma}^2T}{b}
    \nonumber\\
    & \rightarrow &0, \ \ \mbox{as}\  b \rightarrow  \infty.
\end{eqnarray}
Hence, (\ref{eq403-theorem42}), (\ref{eq405-theorem42}),
(\ref{eq407-theorem42}) and (\ref{eq410-theorem42}) yield that
\begin{eqnarray*}
  \lim \limits_{b\rightarrow \infty}\mathbf{P}[\tau_{b}^{\pi^*_{b}}\leq T]=0.
\end{eqnarray*}
\end{proof}
\setcounter{equation}{0}
\section{{\small {\bf  Numerical examples}} }
\vskip 10pt \noindent In this section, we present  some numerical
examples to give the readers an intuitive impression on the
relations between the results and  model parameters. Setting the
parameters at  suitable level, we  portray how the debt rate
$\delta$,  the  constrained risk level of bankrupt probability, the
initial capital $x$ , the volatility rate $\sigma^2 $ and  the
profit rate $\mu$ impact on the optimal return function and the
optimal dividend strategy based on the PDE (\ref{67}) below. we also
show the figures  of the optimal return function $g(x)$ and the
associated optimal feedback control function $a^*(x)$.
\begin{example}\label{ex51}
Let $\mu=2$, $\sigma^2=50$, $c=0.05$, $\alpha=0.5$, $\beta=8$,
$T=300$ and $b=100$. Figure \ref{g_de} shows that the optimal return
function $g_\delta(x)$ decreases with the debt rate $\delta$.
\end{example}
\begin{figure}[H]
\includegraphics[width=0.7\textwidth]{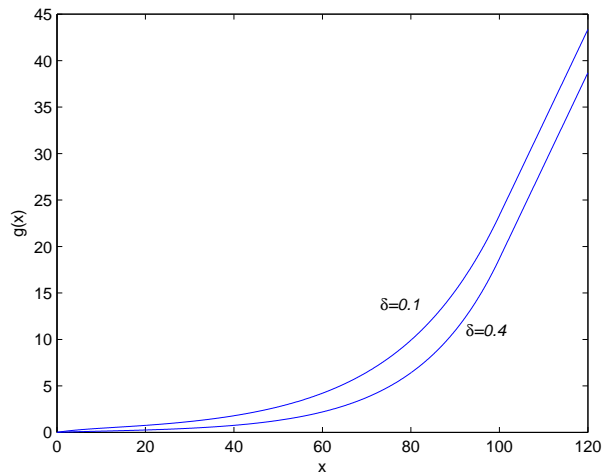}
\caption{ The optimal return function $g_\delta(x)$ as a function of
$\delta$.\ \  (Parameters: $ \mu=2, \sigma=50, c=0.05, \alpha=0.5,
\beta=8, T=300, b=100 $ )}\label{g_de}
\end{figure}
\begin{example}\label{ex52}
Let $\mu=2$, $\sigma^2=50$, $\delta=0.2$, $c=0.05$, $\alpha=0.5$,
$\beta=8$, $T=300$. Let $b(\varepsilon)$ be the solution of
$1-\phi(T,b)=\varepsilon$, where $\phi(T,b)$ is defined in Lemma
\ref{lemma65}. Thus given a  constrained risk level $\varepsilon$ of
bankrupt probability, $b(\varepsilon)$ is the associated dividends
barrier. Figure \ref{b_ep} shows that the dividends barrier
$b(\varepsilon)$ decreases with the constrained  risk level
$\varepsilon$.
\end{example}
\begin{figure}[H]
\includegraphics[width=0.7\textwidth]{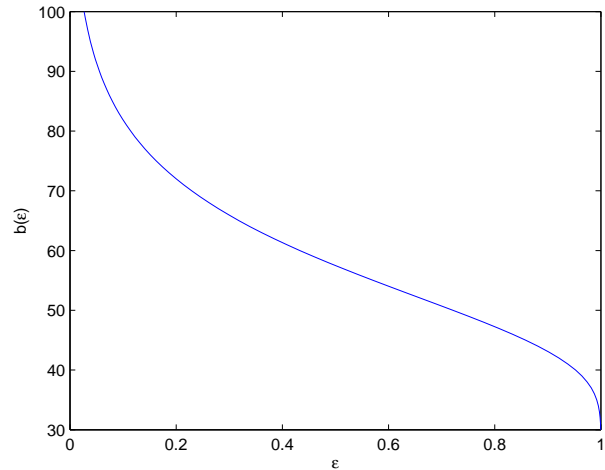}
\caption{ Dividends barrier $b(\varepsilon)$ as a function of
$\varepsilon$.\  (Parameters: $\mu=2, \sigma^2=50, \delta=0.2,
c=0.05, \alpha=0.5, \beta=8, T=300$ )}\label{b_ep}
\end{figure}
\begin{example}\label{ex53}
Let $\mu=2$, $\delta=0.2$, $c=0.05$, $\alpha=0.5$, $\beta=8$,
$T=300$. Let $b_\sigma(\varepsilon)$ be the solution of
$1-\phi(T,b)=\varepsilon$, where $\phi(T,b)$ is defined in Lemma
\ref{lemma65}. We see from Figure \ref{b_si} that at the same
constrained risk level, the bigger the volatility rate $\sigma$ is,
the higher the dividends barrier $b_\sigma(\varepsilon)$ is.
\end{example}
\begin{figure}[H]
\includegraphics[width=0.7\textwidth]{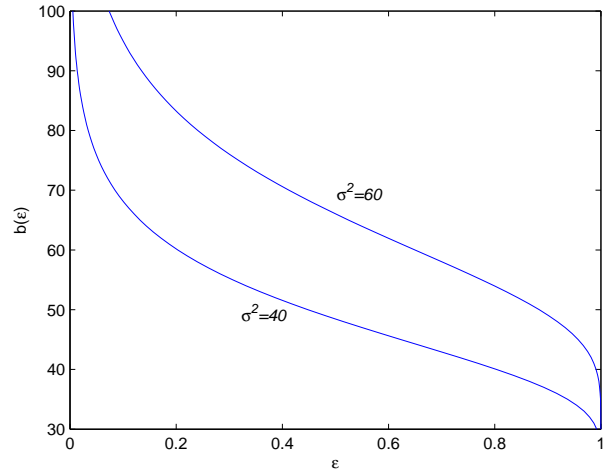}
\caption{ Dividends barrier $b_\sigma(\varepsilon)$ as a function of
$\sigma^2$.\ \ (Parameters: $\mu=2,  \delta=0.2, c=0.05, \alpha=0.5,
\beta=8, T=300$ )}\label{b_si}
\end{figure}
\begin{example}\label{ex54}
Let $\mu=2$, $\sigma^2=50$ $\delta=0.2$, $c=0.05$, $\alpha=0.5$,
$\beta=8$, $T=300$. Let $b_\sigma(\varepsilon)$ be the solution of
$1-\phi(T,b)=\varepsilon$, where $\phi(T,b)$ is defined in Lemma
\ref{lemma65}. $R_0=x$ is the initial reserve and $\varepsilon$ is
the constrained risk  level of bankrupt probability. We see from
figure \ref{x_ep} that the lower the initial reserve $x$ is, the
higher the constrained risk  level $\varepsilon$ is.
\end{example}
\begin{figure}[H]
\includegraphics[width=0.7\textwidth]{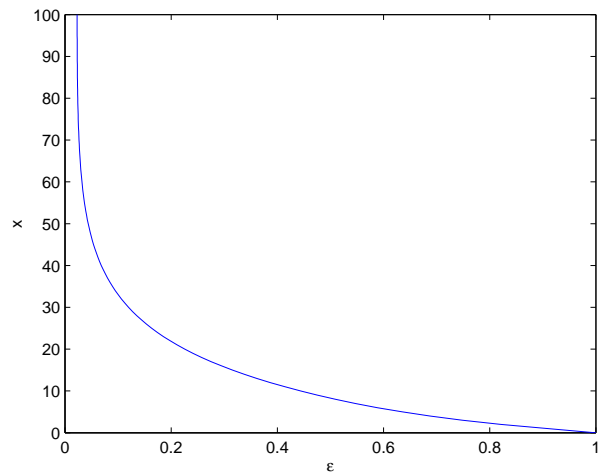}
\caption{ Initial reserve $x(\varepsilon)$ as a function of  the
risk restrained level $\varepsilon$. \ (Parameters: $\mu=2,
\sigma^2=50, \delta=0.2, c=0.05, \alpha=0.5, \beta=8, T=300$
)}\label{x_ep}
\end{figure}
\begin{example}\label{ex55}
Let $\sigma^2=50$, $\delta=0.2$, $c=0.05$, $\alpha=0.5$, $\beta=8$,
$T=300$ and $b=100$. Figure \ref{g_mu} shows that the optimal return
function $g_\mu(x)$ increases with the profit rate $\mu$.
\end{example}
\begin{figure}[H]
\includegraphics[width=0.7\textwidth]{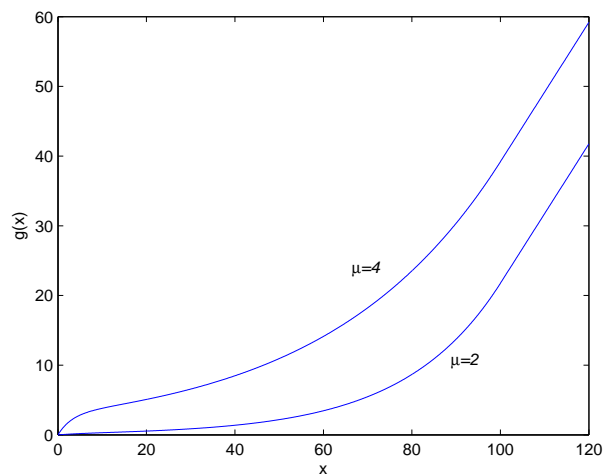}
\caption{ The optimal return function $g_\mu(x)$ as a function of
$\mu$.\ \  (Parameters: $ \sigma^2=50, \delta=0.2, c=0.05,
\alpha=0.5, \beta=8, T=300, b=100 $ )}\label{g_mu}
\end{figure}
\begin{example}\label{ex56}
Let $\mu=2$, $\delta=0.2$, $c=0.05$, $\alpha=0.5$, $\beta=8$,
$T=300$ and $b=100$. Figure \ref{g_si} shows that the optimal return
function $g_\sigma(x)$ increases with the volatility rate
$\sigma^2$.
\end{example}
\begin{figure}[H]
\includegraphics[width=0.7\textwidth]{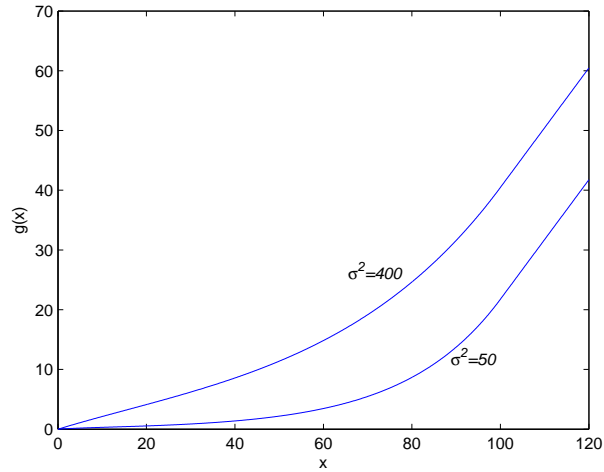}
\caption{ The optimal return function $g_\sigma(x)$ as a function of
$\sigma^2$.\ \  (Parameters: $ \mu=2, \delta=0.2, c=0.05,
\alpha=0.5, \beta=8, T=300, b=100 $ )}\label{g_si}
\end{figure}
\begin{example}\label{ex57}
Let $\mu=2$, $\sigma^2=50$, $\delta=0.2$, $c=0.05$, $\alpha=0.5$,
$\beta=8$, $T=300$ and $b=100$. Set $x_\alpha=4.72$,
$x_\beta=94.79$,  the images of the optimal return function $g(x)$
as well as the optimal feedback control function $a^*(x)$ are as
follows (see Figure \ref{g_x} and Figure \ref{a_x}).
\end{example}
\begin{figure}[H]
\includegraphics[width=0.7\textwidth]{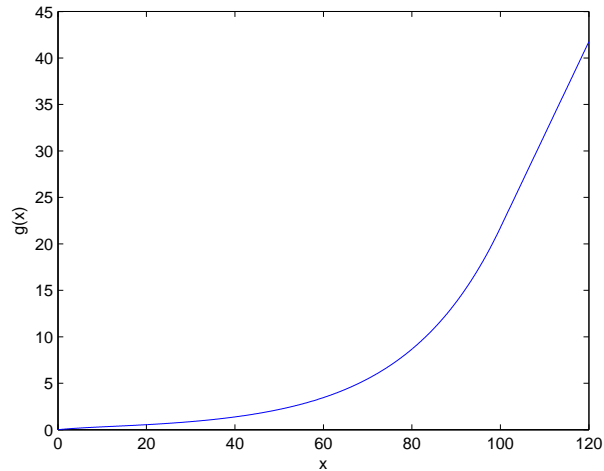}
\caption{ The optimal return function $g(x)$. \ \ (Parameters:
$\mu=2, \sigma^2=50, \delta=0.2, c=0.05, \alpha=0.5, \beta=8, T=300,
b=100 $ )}\label{g_x}
\end{figure}
\begin{figure}[H]
\includegraphics[width=0.7\textwidth]{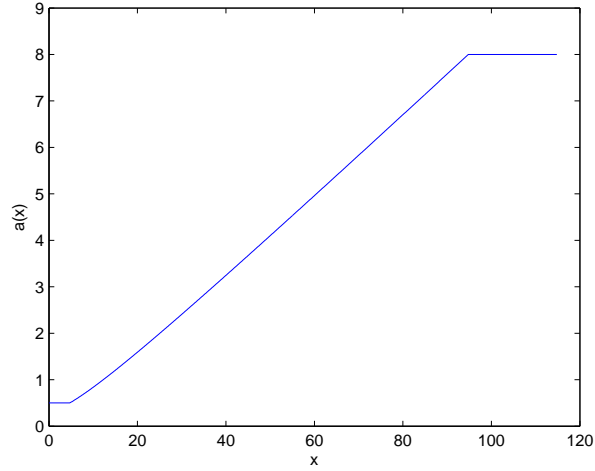}
\caption{ The optimal feedback control function $a(x)$.\ \
(Parameters: $\mu=2, \sigma^2=50, \delta=0.2, c=0.05, \alpha=0.5,
\beta=8, T=300, b=100 $ )}\label{a_x}
\end{figure}
\setcounter{equation}{0}
\section{{{\bf Properties of $V(x,b)$ and Bankrupt Probability }} }
\vskip 5pt \noindent In this section, we will discuss some important
properties of the optimal return function $V(x,b)$ and bankrupt
probability, which are used to prove the main result of this paper.
The rigorous proofs of these properties will be given in the
appendix.  In view of Lemma \ref{lemma81} in the appendix, different
value of $\frac{2\delta}{\mu}$ can lead to three different cases.
When $\frac{2\delta}{\mu}<\alpha$, this case is the most
complicated. We select this case as the basis of our discussion
throughout the paper, and the results of the other two cases are
almost same.
\begin{lemma}\label{lemma61} If $f(x)\in C^{2}$ and
satisfies the following HJB equation and boundary
conditions,
\begin{eqnarray}\label{61}
\left\{
\begin{array}{l l l}
\max\limits_{a\in[\alpha,\beta]}[\frac{1}{2}\sigma^{2}a^{2}f^{''}(x)+
(\mu a-\delta)f^{'}(x)-cf(x)]=0,\ \mbox{for}\  0\leq x\leq b_{0},\\
f^{'}(x)=1, \ \mbox{ for}\   x\geq b_{0},\\
f^{''}(x)=0, \ \mbox{ for} \  x\geq b_{0},\\
f(0)=0,
\end{array}\right.
\end{eqnarray}
then we have
\begin{eqnarray*}
b_0 = \inf \{x\geq 0: f''(x)=0 \}
\end{eqnarray*}
and
\begin{eqnarray*}
\left\{
\begin{array}{l l l}
\max\mathcal {L}f(x)\leq 0\ \mbox{and }\  f^{'}(x)\geq 1 \
\mbox{for}\ x \geq0,\\
f(0)=0,
\end{array}\right.
\end{eqnarray*}
where   $\mathcal {L}= \frac{1}{2}\sigma^{2}a^{2}\frac{d^{2}}{d
x^{2}}+(\mu a-\delta)\frac{d}{d x}-c$.
\end{lemma}

\begin{lemma}\label{lemma62}
Let $b>b_{0}$ be a predetermined variable. If $g\in C^{1}(R_+)$,
$g\in C^{2}(R_+\setminus \{b\})$ and satisfies the following HJB
equation and boundary conditions,
\begin{eqnarray}\label{62}
\left\{
\begin{array}{l l l}
\max\limits_{a\in[\alpha,\beta]}[\frac{1}{2}\sigma^{2}a^{2}g^{''}(x)+
(\mu a-\delta)g^{'}(x)-c g(x)]=0,\ \mbox{for}\ 0\leq x\leq b,\\
g^{'}(x)=1, \ \ \mbox{ for}\ \   x\geq b,\\
g^{''}(x)=0, \ \ \mbox{ for} \ \  x>b,\\
g(0)=0,
\end{array}\right.
\end{eqnarray}
then we have
\begin{eqnarray}\label{63}
\left\{
\begin{array}{l l l}
\max\mathcal {L}g(x)\leq 0,\  \mbox{for}\ \ x \geq
0,\\
g^{'}(x)\geq 1, \ \mbox{for}\ \  x \geq
b,\\
g(0)=0,
\end{array}\right.
\end{eqnarray}
where $b_0$  and $\mathcal {L}$ are the same as in Lemma
\ref{lemma61}, $g''(b):=g''(b-)$. The expression of $g(x)$ can be
written as
\begin{eqnarray}\label{64} g(x,b)=\left\{
   \begin{array}{l l l}
    k_1(e^{r_+(\alpha)}x-e^{r_-(\alpha)x}),\ 0\leq x<x_\alpha,\\
   k_2[\frac{\alpha \mu-2\delta}{2c} + \int_{x_\alpha}^x exp(-\frac{\mu}
     {\sigma^2} \int_{x_\alpha}^y\frac{dv}{a(v)})dy],\ x_\alpha \leq x < x_\beta, \\
   k_3 e^{r_+(\beta)(x-b_0)}+ k_4 e^{r_-(\beta)(x-b_0)}, \ x_\beta \leq x<b,\\
   x-b+g_3(b), \ x\geq b,
   \end{array} \right.
\end{eqnarray}
where $r_{\pm}(x), x_\alpha, x_\beta, k_1, k_2, k_3$ and $k_4$ are
given by (\ref{87}), (\ref{88}), (\ref{815}), (\ref{819}),
(\ref{827}) and (\ref{826}), respectively.
\end{lemma}

\begin{lemma}\label{lemma63} Let $g(x,b)$ be as the same as in Lemma
\ref{lemma62}. Then $\frac{\partial}{\partial b}g(x,b)\leq 0$  holds
for $b\geq b_0$.
\end{lemma}
\begin{lemma}\label{lemma64} The bankrupt probability
 $\mathbf{P}[\tau_b^{\pi^*_b}\leq T]$  is a strictly
decreasing function w.r.t. the dividends barrier $b$ on $ [x_\beta,
D)$,
 $D:=\inf\{b: \mathbf{P}[\tau_b^{\pi^*_b}\leq
T]=0\}$, and $x_\beta$  is defined by (\ref{815}).
\end{lemma}
\noindent From the proof of  Lemma \ref{lemma62}, for each $x\leq
b$, if we define
\begin{eqnarray}\label{65}
   a^*(x):=arg \   \max\limits_{a\in[\alpha,\beta]}
   [\frac{1}{2}\sigma^{2}a^{2}g^{''}(x)+
   (\mu a-\delta)g^{'}(x)-c g(x)],
\end{eqnarray}
then it follows that $a^*(x)$ can be represented as
\begin{eqnarray}\label{66} a^*(x)=\left\{
   \begin{array}{l l l}
   \alpha, \ \   0\leq x<x_\alpha,\\
   a(x), \ \   x_\alpha \leq x < x_\beta, \\
   \beta, \ \ x\geq x_\beta,
   \end{array} \right.
\end{eqnarray}
where $a(x)$ and $x_\alpha, x_\beta$ are specified by (\ref{814}),
(\ref{88}), (\ref{815}), respectively. We now have the following
lemma.
\begin{lemma}\label{lemma65}
 Let $a^*(x)$ be defined by (\ref{66}), and define $\psi^b(T,x):=
 {\bf{P}}[\tau^{\pi^*_b}_x\leq T]$, i.e., $\psi^b(T,x)$ is the bankrupt
 probability when the initial reserve of $\{R^{\pi^*_b}_t\}_{t\geq 0}$ is $x$ and dividends barrier is
 $b$. Let $\phi^{b}(t,y)\in C^1(0,\infty)\cap C^2(0,b)$ and satisfy the
following partial differential equation and the boundary conditions,
\begin{eqnarray}\label{67}
\left\{
\begin{array}{l l l}
  \phi_{t}^{b}(t,x)=\frac{1}{2}[a^{*}(x)]^2\sigma^2\phi_{xx}^{b}(
   t,x)+(a^{*}(x)\mu-\delta)\phi_{x}^{b}(t,x),\\
   \phi^{b}(0,x)=1,\  \mbox{for}\ \  0<x\leq b, \\
   \phi^{b}(t,0)=0,\phi_{x}^{b}(t,b)=0,\ \mbox{for} \ t>0.
\end{array}\right.
\end{eqnarray}
Then $\phi^{b}(T,x)=1-\psi^{b}(T,x)$, i.e., $\phi^{b}(T,x)$ is
probability that the company will survive on $[0, T]$.
 \end{lemma}
\begin{lemma}\label{lemma66}
Let  $\phi^{b}(t,x)$ solve the equation(\ref{67}). Then
$\phi^{b}(T,b)$ is continuous with respect to the dividends barrier
$b$ on $[b_0, +\infty)$.
\end{lemma}
 \setcounter{equation}{0}
\section{\bf Proof of Main Result}
\vskip 5pt\noindent  In this section, we prove the main result of
this paper, which is described in Theorem \ref{theorem31}. In order
to do this, we first  need the following.
\begin{theorem}\label{theorem71}
Let $a^*_b(x)$ be defined by (\ref{66}), and $f(x)$, $g(x,b)$ be as
the same as in Lemma \ref{lemma61} and  Lemma \ref{lemma62},
respectively.
 Then\\
(i) If $b\leq b_0$, we have  $V(x,b)=V(x,b_0)=V(x)=f(x)$, the
optimal policy associated with $V(x)$ is  $\pi_{b_o}^\ast=\{
 a^*_{b_0}(R^{\pi_{b_o}^\ast}_\cdot), L^{\pi_{b_o}^\ast}_\cdot\}$,
 where the process $\{R^{\pi_{b_o}^\ast}_t, L^{\pi_{b_o}^\ast}_t \}$
 is uniquely determined by the following SDE,
\begin{eqnarray}\label{71}
\left\{
\begin{array}{l l l}
dR_t^{\pi_{b_o}^\ast}=(\mu a^*_{b_0}(R^{\pi_{b_o}^\ast}_t)-\delta
)dt+\sigma a^*_{b_0}(R^{\pi_{b_o}^\ast}_t)d
{W}_{t}-dL_t^{\pi_{b_o}^\ast},\\
R_0^{\pi_{b_o}^\ast}=x,\\
0\leq R^{\pi_{b_o}^\ast}_t\leq  b_0,\\
\int^{\infty}_0 I_{\{t: R^{\pi_{b_o}^\ast}_t
<b_0\}}(t)dL_t^{\pi_{b_o}^\ast}=0.
\end{array}\right.
\end{eqnarray}
(ii) If $b> b_0$, we have $V(x,b)=g(x,b)$ and the optimal policy
$\pi_b^\ast$ is $
 \{a^\ast_b(R^{\pi_b^\ast}_t),L^{\pi_b^\ast}_t\}$, where
$\{R^{\pi_b^\ast}_t, L^{\pi_b^\ast}_t  \} $ is uniquely determined
by the following SDE
\begin{eqnarray}\label{72}
\left\{
\begin{array}{l l l}
dR_t^{\pi_b^\ast}=(\mu a^*_b(R^{\pi_b^\ast}_t)-\delta )dt+\sigma
a^*_b(R^{\pi_b^\ast}_t)d
{W}_{t}-dL_t^{\pi_b^\ast},\\
R_0^{\pi_b^\ast}=x,\\
0\leq R^{\pi_b^\ast}_t\leq  b,\\
\int^{\infty}_0 I_{\{t: R^{\pi_b^\ast}_t
<b\}}(t)dL_t^{\pi_b^\ast}=0.
\end{array}\right.
\end{eqnarray}
\end{theorem}
\begin{proof}
(i)\  If $b\leq b_0$, since $\pi^*_{b_0}\in \Pi_{b_0}\subset\Pi_b $,
we have $V(x, b_0)\leq V(x,b)\leq V(x)$. It suffices to show
$V(x)\leq f(x)=V(x, b_0)$. Since its proof is similar to \cite{t1},
we omit it here. \vskip 10pt \noindent (ii)\ If $b\geq b_0$, denote
$g(x,b)$ by $g(x)$ for simplicity, for any admissible policy
$\pi=\{a_\pi, L^\pi\}$, we assume that $( R^\pi_t, L^\pi_t)$ is the
process (\ref{203}) associated with $\pi $. Let
$\Lambda=\{s:L_{s-}^{\pi}\neq L_{s}^{\pi}\}$,
$\hat{L}=\sum_{s\in\Lambda, s\leq t}(L_{s}^{\pi}-L_{s-}^{\pi})$ be
the discontinuous part of $L_{s}^{\pi}  $ and
$\tilde{L}_{t}^{\pi}=L_{t}^{\pi}-\hat{L}_{t}^{\pi}$ be the
continuous part of $L_{s}^{\pi}$.   Applying generalized It\^{o}
formula to $e^{-c(t\wedge \tau^\pi_x)}g(R_{t\wedge
\tau^\pi_x}^{\pi})$, we have
\begin{eqnarray}\label{73}
e^{-c(t\wedge \tau^\pi_x)}g(R_{t\wedge \tau^\pi_x}^{\pi})&=
&g(x)+\int_{0}^{t\wedge\tau^\pi_x}e^{-cs}\mathcal
{L} g(R_{s}^{\pi})ds\nonumber\\
&&+\int_{0}^{t\wedge\tau^\pi_x}a_\pi\sigma
e^{-cs}g^{'}(R_{s}^{\pi})d\mathcal {W}_{s}\mathcal
-\int_{0}^{t\wedge\tau^\pi_x}e^{-cs}g^{'}(R_{s}^{\pi})dL_{s}^{\pi}\nonumber\\
&&+\sum\limits_{s\in\Lambda ,s\leq t\wedge
\tau^\pi_x}e^{-cs}[g(R_{s}^{\pi})-g(R_{s-}^{\pi})\nonumber\\
&&-g^{'}(R_{s-}^{\pi})(R_{s}^{\pi}-R_{s-}^{\pi})]\nonumber \\
&=& g(x)+\int_{0}^{t\wedge\tau^\pi_x}e^{-cs}\mathcal
{L}  g(R_{s}^{\pi})ds\nonumber\\
&&+\int_{0}^{t\wedge\tau^\pi_x}a_\pi \sigma
e^{-cs}g^{'}(R_{s}^{\pi})d\mathcal {W}_{s}\mathcal
-\int_{0}^{t\wedge\tau^\pi_x}e^{-cs}g^{'}(R_{s}^{\pi})d\tilde{L}_{s}^{\pi}\nonumber\\
&&+\sum\limits_{s\in\Lambda ,s\leq t\wedge
\tau^\pi_x}e^{-cs}[g(R_{s}^{\pi})-g(R_{s-}^{\pi}))],
\end{eqnarray}
where
\begin{eqnarray*}
\mathcal {L}=\frac{1}{2}a^{2}\sigma^{2}\frac{d^{2}}{dx^{2}}+(\mu
a-\delta)\frac{d}{dx}-c.
\end{eqnarray*}
In view of the HJB equation (\ref{62}), $\mathcal {L}
g(R_{s}^{\pi})$ is always non-positive, so is the second term on the
right hand side of(\ref{73}). By taking mathematical expectations at
both sides of (\ref{73}), we get
\begin{eqnarray}\label{74}
{\bf E}\big [e^{-c(t\wedge \tau ^\pi _x)}g(R_{t\wedge \tau ^\pi
_x}^{\pi})\big ]&\leq& g(x)-{\bf E}\big [\int_{0}^{t\wedge \tau ^\pi
_x}e^{-cs}
g^{'}(R_{s}^{\pi})d\tilde{L}_{s}^{\pi}\big ]\nonumber\\
&&+{\bf E}\big[\sum\limits_{s\in\Lambda ,s\leq t\wedge
\tau ^\pi _x}e^{-cs}[g(R_{s}^{\pi})-g(R_{s-}^{\pi})]\big ].\nonumber\\
\end{eqnarray}
Since $g^{'}(x)\geq 1$, for $x\geq b$,
\begin{eqnarray}\label{75}
g(R_{s}^{\pi})-g(R_{s-}^{\pi})\leq-(L_{s}^{\pi}-L_{s-}^{\pi}),
\end{eqnarray}
 which, together with (\ref{74}), implies
that
\begin{eqnarray}\label{76}
{\bf E}\big[e^{-c(t\wedge \tau ^\pi _x)}g(R_{t\wedge \tau ^\pi
_x}^{\pi})\big]&+&{\bf E}\big[\int_{0}^{ t\wedge \tau ^\pi
_x}e^{-cs}dL_{s}^{\pi}\big ]\leq g(x).
\end{eqnarray}
By the definition of $\tau ^\pi _x $ and $g(0)=0$, letting $t
\rightarrow \infty$ in (\ref{76}), we get
\begin{eqnarray}\label{77}
\liminf\limits_{t\rightarrow\infty}e^{-c(t\wedge \tau ^\pi
_x)}g(R_{t\wedge
\tau ^\pi _x}^{\pi})&=&e^{-c\tau}g(0)I_{\{\tau ^\pi _x<\infty\}}\nonumber\\
&&+ \liminf\limits_{t\rightarrow\infty}e^{-ct}g(R_{t})I_{\{\tau ^\pi
_x =\infty\}}\nonumber\\ &\geq &0.
\end{eqnarray}
We deduce from (\ref{76}) and (\ref{77}) that
\begin{eqnarray*}
J(x,\pi)={\bf E}\big[\int_{0}^{\tau ^\pi _x} e^{-cs}dL_{s}^{\pi}
\big]\leq g(x).
\end{eqnarray*}
So
\begin{eqnarray}
V(x,b)\leq g(x).
\end{eqnarray}\label{78}
If we choose the control policy
$\pi_{b}^\ast=\{a^*_b(R^{\pi_{b}^\ast}_\cdot),L^{\pi_{b}^\ast}_\cdot\}$,
  which is uniquely determined by SDE (\ref{72}),
then all the inequalities above become equalities. Hence
\begin{eqnarray*}
V(x,b)=g(x).
\end{eqnarray*}
So we have
\begin{eqnarray}\label{79}
V(x,b)=g(x,b).
\end{eqnarray}
\end{proof}
\vskip 5pt\noindent
Now we prove the main result of this paper.
 \vskip 10pt\noindent {\bf Proof of Theorem
\ref{theorem31}}. \vskip 10pt \noindent If
$\mathbf{P}[\tau_{b_0}^{\pi^*_{b_0}}\leq T]\leq\varepsilon$, the
bankrupt probability constraint does not work and it turns to a
usual optimal control problem, thus the conclusion is obvious.
\vskip 10pt\noindent If $ \mathbf{P}[\tau_{b_0}^{\pi^*_{b_0}}\leq
T]>\varepsilon $, then by lemmas \ref{lemma64}-\ref{lemma66} there
exists a unique $b^\ast$ solving the equation
$\mathbf{P}\{\tau_{b}^{\pi^*_{b}}\leq T\}=\varepsilon $. Moreover,
$b^\ast=\inf\{ b:  b\in \mathfrak{B} \}>b_0$. By Theorem
\ref{theorem71} and Lemma \ref{lemma63}, $V(x, b)=g(x,b)$ for
$b>b_0$ and $V(x, b)$ is decreasing w.r.t. $b$. Therefore, we know
that $b^\ast$ meets (\ref{eq302}) and (\ref{eq303}). So the optimal
policy associated with the optimal return function $ V(x,b^*)$ is $
\{a^\ast_{b^*}(R^{\pi_{b^*}^\ast}_t),L^{\pi_{b^*}^\ast}_t\}$, which
is uniquely determined  by  SDE (\ref{eq304}).
 \vskip 5pt \noindent Due to the fact $b^*>b_0$,
 the inequity (\ref{eq305}) is a direct consequence
of Lemma \ref{lemma63}. Thus we complete the proof. \ \ $\Box$
 \setcounter{equation}{0}
\section{\bf Appendix}
\vskip 5pt\noindent In this section, we first discuss some useful
arguments, then we give the proofs of the lemmas used in the
previous sections. \vskip 5pt \noindent Due to the mathematical
model presented by (\ref{202}), $a(t)$ is required to take values in
the interval $\alpha,\beta$, where $0<\alpha<\beta<+\infty$. Thus,
$a(t)\mu-\delta$ may be negative because $\delta>0$. If
$\beta\mu\leq\delta$, there exists a trivial solution to the
corresponding HJB equations, which has been proved by Choulli,
Taksar, and Zhou\cite{t1}. In the next section, we always assume
that $\beta\mu\geq\delta$. Then the following statements are valid.
\begin{lemma}\label{lemma81} Let $\beta\mu>0$. Then\\
(i) $\frac{2\delta}{\mu}<\alpha$ if and only if $a(0)<\alpha$. In
this case,
\begin{eqnarray}\label{81}
    a(0)=\frac{\mu\alpha^2}{2(\mu\alpha-\delta)}.
\end{eqnarray}
(ii) $\alpha\leq\frac{2\delta}{\mu}<\beta$ if and only if
$\alpha\leq a(0)<\beta$.  In this case,
\begin{eqnarray}\label{82}
a(0)=\frac{2\delta}{\mu}.
\end{eqnarray}
 (iii) $\beta\leq \frac{2\delta}{\mu}$ if
and only if $a(0)\geq \beta$. In this case,
\begin{eqnarray}\label{83}
a(0)=\frac{\mu\beta^2}{2(\mu\beta-\delta)}.
\end{eqnarray}
\end{lemma}
\begin{proof} See Choulli, Taksar and Zhou\cite{t1} for details.
\end{proof}
\noindent Due to Lemma \ref{lemma81}, there are three different
cases to investigate. Since the proof of each case is similar, we
only give  sketch proofs of lemmas in case (i). Thus we suppose
$\frac{2\delta}{\mu}<\alpha$ throughout the following procedure to
prove these lemmas. \vskip 10pt \noindent
 {\bf Proof of lemma
\ref{lemma61}}. \ The complete proof is given in Choulli, Taksar and
Zhou\cite{t1}(2003). \ \ $\Box$

 \vskip 10pt\noindent
{\bf Proof of lemma \ref{lemma62}}. \vskip 10pt \noindent  {\bf Step
1}. For each $x\geq 0$ and $a\geq 0$, define
\begin{eqnarray}\label{84}
 h(x,a)=\frac{1}{2}\sigma^{2}a^{2}g^{''}(x)+
 (\mu a-\delta)g^{'}(x)-c g(x).
\end{eqnarray}
Then, by differentiating $h(x,a)$ w.r.t. $a$, we get the maximizing
function of $h(x,a)$
\begin{eqnarray}\label{85}
  a(x)=-\frac{\mu g^{'}(x)}{\sigma^2 g^{''}(x)}, \ \ x\geq 0.
\end{eqnarray}
In view of Lemma \ref{lemma81} (i), $a(x)\leq \alpha$ for all $x$ in
the right neighborhood of 0. Substituting $a=\alpha$ into
(\ref{62}), and solving the resulting second-order linear ODE, we
get
\begin{eqnarray}\label{86}
 g(x)=k_1(e^{r_+(\alpha)}x-e^{r_-(\alpha)x}),\ \  0\leq x< x_\alpha,
\end{eqnarray}
where $k_1$ and $x_\alpha$ are to be determined and for $x>0$
\begin{eqnarray}\label{87}
\left \{
 \begin{array}{l l l}
 r_+(x) = \frac{-(\mu x-\delta)+\sqrt{(\mu x-\delta)^2+2\sigma^2 c x^2}}
 {\sigma^2 x^2}\\
 r_-(x) = \frac{-(\mu x-\delta)-\sqrt{(\mu x-\delta)^2+2\sigma^2 c x^2}}
 {\sigma^2 x^2}.
 \end{array} \right.
\end{eqnarray}
Due to (\ref{85}) and (\ref{86}), for $x>0$
\begin{eqnarray*}
  a^{'}(x)=\frac{-\mu r_+(\alpha) r_-(\alpha)(r_+(\alpha)-r_-(\alpha))^2
         e^{(r_+(\alpha)+r_-(\alpha))x}}{\sigma^2 (g^{''}(x))^2}>0.
\end{eqnarray*}
Therefore $a(x)$ increases to $\alpha$ at the point $x_\alpha$ given
by
\begin{eqnarray}\label{88}
  x_\alpha=\frac{1}{r_+(\alpha)-r_-(\alpha)}log(\frac{r_-(\alpha)(\mu +
  \alpha \sigma^2 r_-(\alpha))}{r_+(\alpha)(\mu +\alpha \sigma^2 r_+(\alpha))})
  >0.
\end{eqnarray}
\vskip 10pt \noindent  {\bf Step 2}. In view of Proposition 8 in
\cite{t1}, $\alpha \leq a(x)\leq \beta$ in the right neighborhood of
$x_\alpha$. From (\ref{85}), we get
\begin{eqnarray}\label{89}
  g^{''}(x)=-\frac{\mu g^{'}(x)}{\sigma^2 a(x)}.
\end{eqnarray}
Substituting (\ref{89}) into (\ref{62}), differentiating the
resulting equation, and using (\ref{89}) again, we obtain
\begin{eqnarray}\label{810}
  a^{'}(x)=\frac{\mu^2+2 c \sigma^2}{\mu \sigma^2} (1-\frac{u}{a(x)}),
\end{eqnarray}
with
\begin{eqnarray}\label{811}
  u \equiv \frac{2\delta \mu}{\mu^2+2 c^2 \sigma^2}.
\end{eqnarray}
Integrating (\ref{810}), we get
\begin{eqnarray}\label{812}
  G(a(x))= \frac{\mu^2+2 c \sigma^2}{\mu \sigma^2}(x-x_\alpha)+G(\alpha),
\end{eqnarray}
where
\begin{eqnarray}\label{813}
  G(z)=z+u log(z-u).
\end{eqnarray}
Therefore
\begin{eqnarray}\label{814}
  a(x)=G^{-1}(\frac{\mu^2+2 c \sigma^2}{\mu \sigma^2}(x-x_\alpha)+G(\alpha)).
\end{eqnarray}
Obviously, $a(x)$ is increasing. Let $a(x_\beta)=\beta$, we get
\begin{eqnarray}\label{815}
  x_\beta & = & \frac{\mu \sigma^2}{\mu^2+2 c \sigma^2}[G(\beta)-G(\alpha)]+x_\alpha \nonumber\\
          & = & \frac{\mu \sigma^2}{\mu^2+2 c \sigma^2}(\beta-\alpha)+
          \frac{\mu \sigma^2 u}{\mu^2+2 c \sigma^2}log(\frac{\beta-u}{\alpha-u})+x_\alpha.
\end{eqnarray}
Solving (\ref{89}),(\ref{813}) and (\ref{814}), we obtain
\begin{eqnarray}\label{816}
  g(x)=g(x_\alpha)+g^{'}(x_\alpha)\int_{x_\alpha}^x exp(-\frac{\mu}{\sigma^2}
   \int_{x_\alpha}^y\frac{dv}{a(v)})dy,\ \ x_\alpha \leq x < x_\beta,\nonumber\\
\end{eqnarray}
where $g(x_\alpha)$ and $g^{'}(x_\alpha)$ are free constants to be
determined. From (\ref{86}) and (\ref{88}), we deduce
\begin{eqnarray}\label{817}
  g(x_\alpha)=\frac{\alpha \mu-2\delta}{2c} g^{'}(x_\alpha)
\end{eqnarray}
Let
\begin{eqnarray}\label{818}
  k_2 \equiv g^{'}(x_\alpha),
\end{eqnarray}
Then (\ref{86}) and (\ref{817}) imply
\begin{eqnarray}\label{819}
  k_1= \frac{\alpha \mu-2\delta}{2c(e^{r_+(\alpha)x_\alpha}-e^{r_-(\alpha)x_\alpha})} k_2
\end{eqnarray}
Substituting (\ref{817}) and (\ref{818}) into (\ref{816}), we get
\begin{eqnarray}\label{820}
  g(x)=k_2[\frac{\alpha \mu-2\delta}{2c} + \int_{x_\alpha}^x exp(-\frac{\mu}{\sigma^2}
   \int_{x_\alpha}^y\frac{dv}{a(v)})dy],\ x_\alpha \leq x < x_\beta.
\end{eqnarray}
\vskip 10pt \noindent  {\bf Step 3}. In view of Proposition 9 in
\cite{t1}, $a(x)\geq \beta$ holds for $x\geq x_\beta$. Substituting
$a=\beta$ into (\ref{62}), and solving it, we get the following
solution
\begin{eqnarray}\label{821}
  g(x)=k_3 e^{r_+(\beta)(x-b_0)}+ k_4 e^{r_-(\beta)(x-b_0)}, \ x_\beta \leq x<b,
\end{eqnarray}
where $k_3, k_4$ are free constants to be determined and
$r_{\pm}(\beta)$ are given by (\ref{87}). For $x\geq b$, the
solution has the following form
\begin{eqnarray}\label{822}
  g(x)=x-b+k_3 e^{r_+(\beta)(b-b_0)}+ k_4 e^{r_-(\beta)(b-b_0)}, \ x\geq b.
\end{eqnarray}
Next we apply the principle of smooth fit to determine the unknown
constants above. Note that
\begin{eqnarray}\label{823}
\left\{
\begin{array}{l l l}
  g(x_\beta -) = g(x_\beta +) \\
  g^{'}(x_\beta -) = g^{'}(x_\beta +),
\end{array}\right.
\end{eqnarray}
we arrive at
\begin{eqnarray}\label{824}
\left\{
\begin{array}{l l l}
  k_2 \xi = k_3 e^{r_+(\beta)(x_\beta-b_0)} + k_4 e^{r_-(\beta)(x_\beta-b_0)}\\
  k_2 \eta = k_3 r_+(\beta) e^{r_+(\beta)(x_\beta-b_0)} +
             k_4 r_-(\beta) e^{r_-(\beta)(x_\beta-b_0)},
\end{array}\right.
\end{eqnarray}
where
\begin{eqnarray}\label{825}
\left\{
\begin{array}{l l l}
  \xi = \frac{\alpha \mu-2\delta}{2c} + \int_{x_\alpha}^{x_\beta}
         exp(-\frac{\mu}{\sigma^2}\int_{x_\alpha}^y\frac{dv}{a(v)})dy \\
  \eta = exp(-\frac{\mu}{\sigma^2}\int_{x_\alpha}^{x_\beta}\frac{dv}{a(v)}).
\end{array}\right.
\end{eqnarray}
Solving (\ref{824}) for $k_3$ and $k_4$, we get
\begin{eqnarray}\label{826}
\left\{
\begin{array}{l l l}
  k_3 = \frac{\eta-\xi r_-(\beta)}{(r_+(\beta)-r_-(\beta))
         e^{r_+(\beta)(x_\beta-b_0)}} k_2 \equiv A k_2\\
  k_4 = \frac{\xi r_+(\beta)-\eta}{(r_+(\beta)-r_-(\beta))
         e^{r_-(\beta)(x_\beta-b_0)}} k_2 \equiv B k_2.
\end{array}\right.
\end{eqnarray}
Substituting (\ref{826}) into (\ref{821}) and using $g^{'}(b-)=1$,
we obtain
\begin{eqnarray}\label{827}
  k_2 = \frac{1}{A r_+(\beta)e^{r_+(\beta)(b-b_0)}+
        B r_-(\beta)e^{r_-(\beta)(b-b_0)}}.
\end{eqnarray}
Thus, $k_1, k_2, k_3, k_4$ are determined by (\ref{819}),(\ref{826})
and (\ref{827}). We claim that
\begin{eqnarray}\label{828}
   g^{''}(b_-)\geq 0.
\end{eqnarray}
In order to prove this statement, we consider $f(x)$ in Lemma
\ref{lemma61} and notice that $A$ and $B$ in (\ref{826}) have the
same expression both in $f(x)$ and $g(x)$. Since $f^{'}(b_0)=1$ and
$f^{''}(b_0)=0$,
\begin{eqnarray}\label{829}
\left\{
\begin{array}{l l l}
   k^f_2 (A r_+(\beta)+ B r_-(\beta)) = 1\\
   k^f_2 (A r^2_+(\beta)+ B r^2_-(\beta)) = 0,
\end{array}\right.
\end{eqnarray}
where $k^f_2$ is the corresponding constant in Lemma \ref{lemma61}.
From (\ref{829}), we know that $A<0, B>0$ due to $r_+(\beta)>0,
r_-(\beta)<0)$
 and $k^f_2>0$ in $f(x)$. In addition, if we let
\begin{eqnarray}\label{830}
   l(b)\equiv g^{''}_3(b-) = \frac{A r^2_+(\beta)e^{r_+(\beta)(b-b_0)}+
        B r^2_-(\beta)e^{r_-(\beta)(b-b_0)}}{A r_+(\beta)e^{r_+(\beta)(b-b_0)}+
        B r_-(\beta)e^{r_-(\beta)(b-b_0)}},
\end{eqnarray}
then,
\begin{eqnarray}\label{831}
  \frac{\partial l}{\partial b}& = &\frac{A B(r^3_+(\beta)r_-(\beta)+
       r_+(\beta)r^3_-(\beta)-2 r^2_+(\beta)r^2_-(\beta)) e^{(r_+(\beta)+r_(\beta))
       (b-b_0)}}{(A r_+(\beta)e^{r_+(\beta)(b-b_0)}+
        B r_-(\beta)e^{r_-(\beta)(b-b_0)})^2} \nonumber \\
        & > & 0 \nonumber \\
\end{eqnarray}
holds for $A>0, B<0, r_+(\beta)>0$ and $r_-(\beta)<0$. Since
$b>b_0$, we conclude that
\begin{eqnarray}\label{832}
  g^{''}(b-) = l(b) > l(b_0) = \tilde{g}^{''}(b_0) = 0,
\end{eqnarray}
where $\tilde{g}(x)$ is the solution of (\ref{62}) with replacing
$b$ by $b_0$. \vskip 10pt \noindent  {\bf Step 4}.
 Now we only need to prove the solution $g(x)$
satisfies (\ref{62}). It suffices to prove the following conditions,
\begin{eqnarray*}
  \max\limits_{a\in[\alpha,\beta]}[\frac{1}{2}\sigma^{2}a^{2}g^{''}(x)+
  (\mu a-\delta)g^{'}(x)-c g(x)]=0,\ \mbox{for}\ x \geq b.
\end{eqnarray*}
By (\ref{62}), (\ref{832}) and noticing that $g^{'}(b-)=1$, we get
\begin{eqnarray}\label{833}
  \max\mathcal {L}g(x)& = & \mu a - \delta - c (x-b+g(b)) \nonumber \\
                      & \leq & \mu a - \delta - c g(b) \nonumber \\
                      & \leq & \frac{1}{2}\sigma^2 a^2 g^{''}(b-)
                      + (\mu a - \delta)g^{'}_3(b-) - c g(b)\nonumber \\
                      & \leq & 0.
\end{eqnarray}
Thus, we complete the proof.  We summarize the solution as follows.
For $\frac{2\delta}{\mu} < \alpha$,
\begin{eqnarray}\label{834} g(x)=\left\{
   \begin{array}{l l l}
    k_1(e^{r_+(\alpha)x}-e^{r_-(\alpha)x}),\ \ 0\leq x<x_\alpha,\\
  k_2[\frac{\alpha \mu-2\delta}{2c} + \int_{x_\alpha}^x \exp(-\frac{\mu}
     {\sigma^2} \int_{x_\alpha}^y\frac{dv}{a(v)})dy],\ \ x_\alpha \leq x < x_\beta, \\
   k_3 e^{r_+(\beta)(x-b_0)}+ k_4 e^{r_-(\beta)(x-b_0)}, \ \ x_\beta \leq x<b,\\
   x-b+k_3 e^{r_+(\beta)(b-b_0)}+ k_4 e^{r_-(\beta)(b-b_0)}, \ \ x\geq b,
   \end{array} \right.\nonumber \\
\end{eqnarray}
where $r_{\pm}(x), x_\alpha, x_\beta, k_1, k_2, k_3$ and $k_4$ are
defined by (\ref{87}), (\ref{88}), (\ref{815}), (\ref{819}),
(\ref{827}) and (\ref{826}), respectively. \ \ $\Box$
 \vskip
10pt\noindent {\bf Proof of lemma \ref{lemma63}}. \ For $b\geq b_0,
A<0,B>0$, together with (\ref{830}) and (\ref{831}), we have
\begin{eqnarray*}
  \frac{\partial }{\partial b}g(b,x)&=& -\frac{(\alpha \mu-2\delta)}
    {2 c(e^{r_+(\alpha)x_\alpha-r_-(\alpha)x_\alpha})}\\
    &&\times\frac{A r^2_+(\beta)e^{r_+(\beta)(b-b_0)}+B r^2_-(\beta)e^{r_-(\beta)(b-b_0)}}
    {(A r_+(\beta)e^{r_+(\beta)(b-b_0)}+B r_-(\beta)e^{r_-(\beta)(b-b_0)})^2}\\
    & \leq & 0, \ \ \ 0\leq x<x_\alpha; \\
 \frac{\partial }{\partial b}g(b,x)&=&
    -(\frac{\alpha \mu-2\delta}{2c} + \int_{x_\alpha}^x exp(-\frac{\mu}{\sigma^2}
   \int_{x_\alpha}^y\frac{dv}{a(v)})dy)\\
   &&\times\frac{A r^2_+(\beta)e^{r_+(\beta)(b-b_0)}+B r^2_-(\beta)e^{r_-(\beta)(b-b_0)}}
   {(A r_+(\beta)e^{r_+(\beta)(b-b_0)}+B r_-(\beta)e^{r_-(\beta)(b-b_0)})^2}\\
   & \leq & 0, \ \ \ x_\alpha \leq x < x_\beta; \\
 \frac{\partial }{\partial b}g(b,x)&=&
   (A e^{r_+(\beta)(b-b_0)}+B e^{r_-(\beta)(b-b_0)})\\
   &&\times\frac{(A r^2_+(\beta)e^{r_+(\beta)(b-b_0)}+B r^2_-(\beta)e^{r_-(\beta)(b-b_0)})}
   {(A r_+(\beta)e^{r_+(\beta)(b-b_0)}+B r_-(\beta)e^{r_-(\beta)(b-b_0)})^2}\\
   & \leq & 0, \ \ \ x_\beta \leq x<b;\\
 \frac{\partial }{\partial b}g(b,x)&=&
   (A e^{r_+(\beta)(b-b_0)}+B e^{r_-(\beta)(b-b_0)})\\
   &&\times\frac{(A r^2_+(\beta)e^{r_+(\beta)(b-b_0)}+B r^2_-(\beta)e^{r_-(\beta)(b-b_0)})}
   {(A r_+(\beta)e^{r_+(\beta)(b-b_0)}+B r_-(\beta)e^{r_-(\beta)(b-b_0)})^2}\\
   & \leq & 0, \ \ \ x\geq b.
\end{eqnarray*}
Thus, the proof is completed.\ \ $\Box$
 \vskip 10pt\noindent {\bf
Proof of lemma \ref{lemma64}}.\ We can prove that
$\mathbf{P}[\tau_b^{\pi^*_b}\leq T]$ is decreasing in $b$ along the
lines of Theorem 3.1 in \cite{ime02}(2008). Here we only need to
prove that $\mathbf{P}[\tau_b^{\pi^*_b}\leq T]$ is strictly
decreasing in $b$ on $[x_\beta, D)$. We denote
$\mathbf{P}[\tau_b^{\pi^*_b}\leq T]$ by $\mathbf{P}[\tau_b^b\leq T]$
for simplicity. \vskip 10pt\noindent For any $b_1, b_2$ satisfying
$D\geq b_2\geq b_1 \geq x_\beta$, we need to prove that
\begin{eqnarray*}
\mathbf{P}[\tau_{b_1}^{b_1}\leq T]-\mathbf{P}[\tau_{b_2}^{b_2}\leq
T]>0
\end{eqnarray*}
By comparison theorem, we have
$$ \mathbf{P}[\tau_{b_1}^{b_1}\leq T]-\mathbf{P}[\tau_{b_2}^{b_2}\leq T]\geq
\mathbf{P}[\tau_{b_1}^{b_2}\leq T]-\mathbf{P}[\tau_{b_2}^{b_2}\leq
T].
$$
So we the proof can be reduced to proving that
\begin{eqnarray}\label{835}
\mathbf{P}[\tau_{b_1}^{b_2}\leq T]-\mathbf{P}[\tau_{b_2}^{b_2}\leq
T]>0.
\end{eqnarray}
\vskip 10pt\noindent Define stochastic processes $R_t^{(1)}$,
$R_t^{(2)}$, $R_t^{(3)}$, $R_t^{(4)}$ by the following SDEs:
\begin{eqnarray*}
dR_t^{(1)}&=&[\mu
a_{b_2}^*(R_t^{(1)})-\delta]dt+a_{b_2}^*(R_t^{(1)})\sigma d\mathcal
{W}_t-dL_t^{b_2}, R_0^{(1)}=b_1;\\
dR_t^{(2)}&=&[\mu
a_{b_2}^*(R_t^{(2)})-\delta]dt+a_{b_2}^*(R_t^{(2)})\sigma
d\mathcal{W}_t-dL_t^{b_2}, R_0^{(2)}=b_2;\\
dR_t^{(3)}&=&[\mu
a_{b_2}^*(R_t^{(3)})-\delta]dt+a_{b_2}^*(R_t^{(3)})\sigma
d\mathcal{W}_t-dL_t^{b_2}, R_0^{(3)}=\frac{b_1+b_2}{2};\\
dR_t^{(4)}&=&[\mu
a_{b_2}^*(R_t^{[4]})-\delta]dt+a_{b_2}^*(R_t^{(4)})\sigma
d\mathcal{W}_t, R_0^{(4)}=\frac{b_1+b_2}{2},
\end{eqnarray*}
respectively, where $D\geq b_2\geq b_1\geq x_\beta$ and $
a^{*}(\cdot)$ is defined by (\ref{66}). \vskip 10pt \noindent First,
let $\tau^{b_1}=\inf\limits_{t\geq 0}\{t:R_t^{(2)}=b_1\}$, $A= \{
\tau^{b_1}\leq T\}$, $B=\big\{\sl R_{t}^{(2)}$ will go to bankruptcy
in a time interval $[\tau^{b_1},\tau^{b_1}+T]$ and $\tau^{b_1}\leq T
\big \}$,  $D=\{\inf\limits_{0\leq t\leq T}R_t^{(3)}>b_1\}$ and
    $E=\{\inf\limits_{0\leq t\leq T}R_t^{(4)}>b_1,\sup\limits_{0\leq
t\leq T}R_t^{(4)} <b_2\}$. Then
\begin{eqnarray} \label{836}
\{\tau_{b_2}^{b_2}\leq T\}\subset B \subset A.
\end{eqnarray}
 Moreover, by using strong Markov property of
$R_t^{[2]}$, we have
\begin{eqnarray}\label{837}
\mathbf{P}[\tau_{b_1}^{b_2}\leq T]=\mathbf{P}[B|A].
\end{eqnarray}
By comparison theorem on SDE, we have
\begin{eqnarray} \label{838}
\mathbf{P}(A^c)\geq \mathbf{P}(D)\geq \mathbf{P}(E).
\end{eqnarray}
Since $a_{b_2}^*(x)=\beta$ we have
\begin{eqnarray}\label{839}
R_t^{(4)}= \frac{b_1+b_2}{2}+ [\mu\beta-\delta]t+\sigma\mathcal{W}_t
\ \mbox{ on $ E$}.
\end{eqnarray}
We deduce from (\ref{839}) and properties of Brownian motion with
drift (cf. Andrei,Borodin,Paavo,Salminen \cite{s110} (2002)) that
\begin{eqnarray*}
\mathbf{P}(E)=\frac{e^{-\mu'^2T/2}}{\sqrt{2\pi
T}}\sum_{k=-\infty}^{\infty}\int_{b_1/\sigma\beta}^{b_2/\sigma\beta}e^{\mu'(z-x)}
[(e^{-(z-x+\frac{2k(b_2-b_1)}{\sigma\beta})^2/2T})\\
-(e^{-(z+x-\frac{2b_1-2k(b_2-b_1)}{\sigma\beta})^2/2T})]dz>0,\nonumber
\end{eqnarray*}
where $\mu'=(\beta\mu-\delta)/\sigma$ and
$x=\frac{b_1+b_2}{2\sigma\beta}$. Thus we get
\begin{eqnarray}\label{840}
\mathbf{P}(A^c)>0.
\end{eqnarray}
We also know from Theorem \ref{theorem41}  that
$\mathbf{P}[\tau_{b_1}^{b_2}\leq T]\geq
\mathbf{P}[\tau_{b_2}^{b_2}\leq T]>0$, which together with
(\ref{836}), (\ref{837}) and (\ref{840}), implies that
\begin{eqnarray*}
 \mathbf{P}[\tau_{b_1}^{b_2}\leq T]-\mathbf{P}[\tau_{b_2}^{b_2}\leq T]
 &\geq &
\mathbf{P}[\tau_{b_1}^{b_2}\leq T]-\mathbf{P}(B)\nonumber\\
&=& \mathbf{P}[\tau_{b_1}^{b_2}\leq T]-\mathbf{P}(A)\mathbf{P}(B|A)\nonumber\\
&=&\mathbf{P}[\tau_{b_1}^{b_2}\leq T](1-\mathbf{P}(A))\nonumber\\
&=&\mathbf{P}[\tau_{b_1}^{b_2}\leq T]\mathbf{P}(A^c) \nonumber\\
&>& 0.
\end{eqnarray*}
 Thus the proof is completed. $\Box$

\vskip 10pt\noindent {\bf Proof of lemma \ref{lemma65}}. \  Let
$\phi(t,x)\equiv \phi^{b}(t,x)$. Setting $ \tau_{x}^{b}:=
\tau_{x}^{\pi^*_b}$ and applying  the generalized It\^{o}
 formula to $(R^{\pi_{b}^\ast,x}_t, L_t^{\pi_{b}^\ast} )$  and $\phi(t,x)$
 we have for $0<x\leq b$,
\begin{eqnarray}\label{841}
  \phi(T-(t\wedge\tau_{x}^{b}),R^{\pi_{b}^\ast,x}_{t\wedge\tau_{x}^{b}})
  &=&\phi(T,x)\nonumber\\
  &+&\int_{0}^{t\wedge\tau_{x}^{b}}(\frac{1}{2}a^{*2}(R^{\pi_{b}^\ast,x}_{s})
  \sigma^{2}\phi_{yy}(T-s,R^{\pi_{b}^\ast,x}_{s})\nonumber\\
  &+&(a^*_b(R^{\pi_{b}^\ast,x}_{s})\mu-\delta)\phi_{x}(T-s,R^{\pi_{b}^\ast,x}_{s})
  \nonumber\\
  &-&\phi_{t}(T-s,R^{\pi_{b}^\ast,x}_{s}))ds-\int_{0}^{t\wedge\tau_{x}^{b}}
  \phi_{y}(T-s,R^{\pi_{b}^\ast,x}_{s})dL_{s}^{b}\nonumber\\
  &+&\int_{0}^{t\wedge\tau_{x}^{b}}
  a^*(R^{\pi_{b}^\ast,x}_{s})\sigma\phi_{x}(T-s,R^{\pi_{b}^\ast,x}_{s})dW_{s}.\nonumber \\
\end{eqnarray}
Setting $t=T$ and taking mathematical expectation at both sides of
(\ref{841}) yield that
\begin{eqnarray*}
\phi(T,x)&=&\mathbf{E}[\phi(T-(T\wedge\tau_{x}^{b}),
R^{\pi_{b}^\ast,x}_{T\wedge\tau_{x}^{b}}
)]\\
&=&\mathbf{E}[\phi(0, R^{\pi_{b}^\ast,x}_{T})1_{T<\tau_{x}^{b}}]+
\mathbf{E}[\phi(T-\tau_{x}^{b},0)1_{T\geq\tau_{x}^{b}})]\\
&=&\mathbf{E}[1_{T<\tau_{x}^{b}}]\\
&=& 1-\psi(T,x).
\end{eqnarray*}
Thus we complete the proof.\ \ $\Box$

 \vskip 10pt\noindent Define $u(x):=\frac{1}{2}a^{*2}(x)\sigma^2$,
 $v(x):=a^{*2}(x)\mu-\delta$,  the equation (4.13) becomes
\begin{eqnarray}\label{842}
\phi_{t}^{b}(t,x)=u(x)\phi_{xx}^{b}( t,x)+v(x)\phi_{x}^{b}(t,x).
\end{eqnarray}
Obviously, $u(x)$ and $v(x)$ are continuous in $[0,b]$ due to the
fact that $a^{*}(x)$ is continuous w.r.t $x$. Thus there exists a
unique solution in $C^1(0,\infty)\cap C^2(0,b)$ for (\ref{67}).
Moreover, $u^{'}(x)$, $v^{'}(x)$, $u^{''}(x)$ are bounded in
$(0,x_\alpha)$, $(x_\alpha, x_\beta)$, $(x_\beta, b)$, respectively.
Now we are ready to prove that the bankrupt probability
$\psi^{b}(T,b)$ is continuous with respect to the dividends barrier
$b(b\geq b_0)$.
 \vskip 10pt\noindent {\bf Proof of lemma
\ref{lemma66}}. \ It suffices to prove that $\phi^{b}(t,x)$ is
continuous in b. Let $x=by$ and  $\theta^{b}(t,y)=\phi^{b}(t,by)$,
the equation (\ref{67}) becomes
\begin{eqnarray}\label{843}
\left\{
\begin{array}{l l l}
\theta_t^{b}(t,y)=[u(by)/b^2]\theta_{yy}^{b}(t,z)+[v(by)/b]\theta_y^{b}(t,y),\\
\theta^{b}(0,y)=1, \ \mbox{for}\  0<y\leq 1, \\
\theta^{b}(t,0)=0,\theta_{y}^{b}(t,1)=0,\ \mbox{for} \  t>0.
\end{array}\right.
\end{eqnarray}
So the proof of Lemma \ref{lemma66} can be reduced to proving
$\lim\limits_{b_2\rightarrow
b_1}\theta^{b_2}(t,1)=\theta^{b_1}(t,1)$ for fixed $b_1>b_0$.
Setting $w(t,y)=\theta^{b_2}(t,y)-\theta^{b_1}(t,y)$ and noticing
that $\theta^{b}(t,y)$ is continuous at $y=1$ for any $b\geq b_0$,
it suffices to show that
\begin{eqnarray}\label{844}
\int_0^{t}\int_0^1 w^2(s,y)dyds \rightarrow 0,\ \mbox{as}\ b_2
\rightarrow b_1.
\end{eqnarray}
Thus we have
\begin{eqnarray}\label{845}
\left\{
\begin{array}{l l l}
w_t(t,y)&=&[u(b_2y)/b_2^2]w_{yy}(t,y)+[v(b_2y)/b_2]w_{y}(t,y)\\
&+& \{u(b_2y)/b_2^2-u(b_1z)/b_1^2\}\theta_{yy}^{b_1}(t,y)\\
&+&\{u(b_2y)/b_2^2-u(b_1y)/b_1^2\}\theta_y^{b_1}(t,y),\\
w(0,y)&=&0,\ \mbox{for}\ 0<y\leq 1,\\
w(t,0)&=&0,\ w_y(t,1)=0,\ \mbox{for}\ t>0.
\end{array}\right.
\end{eqnarray}
Multiplying the first equation in (\ref{845}) by $w(t,y)$ and then
integrating on $[0, 1]\times[0,t]$,
\begin{eqnarray}\label{846}
 &&\int_0^{t}\int_0^1 w(s,y)w_t(s,y)dyds
  \nonumber\\
  &=&\int_0^{t}\int_0^1 \big\{[u(b_2y)/b_2^2]w(s,y)w_{yy}(s,y)
  \nonumber\\
  &&+\int_0^{t}\int_0^1 [v(b_2y)/b_2]w(s,y)w_y(s,y)\nonumber\\
  &&+\int_0^{t}\int_0^1 [u(b_2y)/b_2^2-u(b_1y)/b_1^2]w(s,y)\theta_{yy}^{b_1}(t,y)
  \nonumber\\
  &&+\int_0^{t}\int_0^1w(s,y)[v(b_2y)/b_2-v(b_1y)/b_1]w(s,y)\theta_y^{b_1}(t,y)\big \}dyds
  \nonumber\\
&\equiv & E_1 +E_2+E_3+E_4.
\end{eqnarray}
We now estimate each terms at both sides of (\ref{846}) as follows.
\begin{eqnarray}\label{847}
  \int_0^{t}\int_0^1 w(s,y)w_t(s,y)dyds&=&\int_0^1 \frac{1}{2}w^2(t,y)dy.
\end{eqnarray}
By property of $a^*(x)$ and the definition of $u(x)$ and $v(x)$
there exit positive constants $D_1 $, $D_2  $ and $ D_3 $ such that
$[v(b_2y)/b_2]^2\leq D_1$, $[u(by)/b^2]'\geq 0$,
$[u(b_2y)/b_2^2]\geq D_2$ and  $[a(b_2y)/b_2^2]'\leq D_3$. For any
$\lambda_1>0$ and $\lambda_2>0$, by these facts and Young's
inequality, we estimate $E_1$ and $E_2$ as follows,
\begin{eqnarray}\label{848}
  E_1&=&\int_0^{t}\int_0^1[u(b_2y)/b_2^2]w(s,y)w_{yy}(s,y)dyds\nonumber\\
    &=& -\int_0^{t}\int_0^1 [u(b_2y)/b_2^2]w_y^2(s,y)dyds\nonumber\\
    &&-\int_0^{t}\int_{x_\alpha/b_2}^{x_\beta/b_2}[u(b_2y)/b_2^2]^{'}w_y(s,y)w(s,y)dyds
    \nonumber\\ &\leq & -D_2\int_0^{t}\int_0^1 w_y^2(s,y)dyds\nonumber\\
    && +D_3\int_0^{t}\int_0^1 [\lambda_1 w_y^2(s,y)+\frac{1}{4\lambda_1}
    w^2(s,y)]dyds,
\end{eqnarray}
\begin{eqnarray}\label{449}
   E_2&=&\int_0^{t}\int_0^1[v(b_2y)/b_2]w(s,y)w_y(s,y)dyds\nonumber\\
      &\leq
      &\lambda_2\int_0^{t}\int_0^1 w_y^2(s,y)dyds\nonumber\\
      &&+\frac{D_1}{4\lambda_2}\int_0^{t}\int_0^1 w^2(s,y)dyds.
\end{eqnarray}
It is easy to see that  $[u(by)/b^2]$, $[u(by)/b^2]'$and $[v(by)/b]$
are Lipschitz continuous for all $y\in(x_\alpha/b_2,x_\beta/b_1)$,
that is, there exists $L>0$ such that
\begin{eqnarray}\label{850}
 &&|[u(b_2y)/b_2^2]-[u(b_1y)/b_1^2]|\leq L|b_2-b_1|,\nonumber \\
 &&|[u(b_2y)/b_2^2]'-[u(b_1y)/b_1^2]'|\leq L|b_2-b_1|,\nonumber\\
 &&|[v(b_2y)/b_2]-[v(b_1y)/b_1]|\leq L|b_2-b_1|.
\end{eqnarray}
 $E_3$ has the following expressions,
\begin{eqnarray}\label{851}
  E_3&=&\int_0^t\int_0^1 \{u(b_2y)/b_2^2-u(b_1y)/b_1^2\}w(s,y)\theta_{yy}^{b_1}(s,y)dyds
   \nonumber\\
  &=&-\int_0^t\int_0^1 [u(b_2y)/b_2^2u(b_1y)/b_1^2]w_y(s,y)\theta_{y}^{b_1}(s,y)dyds
   \nonumber\\
  &&-\int_0^t\int_0^{x_\alpha/b_2}[u(b_2y)/b_2^2-u(b_1y)/b_1^2]'
    w(s,y)\theta_{y}^{b_1}(s,y) dyds \nonumber\\
  &&-\int_0^t\int_{x_\alpha/b_2}^{x_\alpha/b_1}[u(b_2y)/b_2^2-u(b_1y)/b_1^2]'
    w(s,y)\theta_{y}^{b_1}(s,y)dyds \nonumber\\
  &&-\int_0^t\int_{x_\alpha/b_1}^{x_\beta/b_2}[u(b_2y)/b_2^2-u(b_1y)/b_1^2]'
    w(s,y)\theta_{y}^{b_1}(s,y)dyds  \nonumber\\
  &&-\int_0^t\int_{x_\beta/b_2}^{x_\beta/b_1}[u(b_2y)/b_2^2-u(b_1y)/b_1^2]'
    w(s,y)\theta_{y}^{b_1}(s,y)dyds  \nonumber\\
  &&-\int_0^t\int_{x_\beta/b_1}^{1}[u(b_2y)/b_2^2-u(b_1y)/b_1^2]'
    w(s,y)\theta_{y}^{b_1}(s,y)dyds  \nonumber\\
  &=&-\int_0^t\int_0^1 [u(b_2y)/b_2^2u(b_1y)/b_1^2]w_y(s,y)\theta_{y}^{b_1}(s,y)dyds
   \nonumber\\
  &&-\int_0^t\int_{x_\alpha/b_2}^{x_\alpha/b_1}[u(b_2y)/b_2^2-u(b_1y)/b_1^2]'
    w(s,y)\theta_{y}^{b_1}(s,y)dyds \nonumber\\
  &&-\int_0^t\int_{x_\alpha/b_1}^{x_\beta/b_2}[u(b_2y)/b_2^2-u(b_1y)/b_1^2]'
    w(s,y)\theta_{y}^{b_1}(s,y)dyds  \nonumber\\
  &&-\int_0^t\int_{x_\beta/b_2}^{x_\beta/b_1}[u(b_2y)/b_2^2-u(b_1y)/b_1^2]'
    w(s,y)\theta_{y}^{b_1}(s,y)dyds  \nonumber\\
  &=& E_{31}+E_{32}+E_{33}+E_{34}.
\end{eqnarray}
By (\ref{850}) and (\ref{851})and Young's inequality for any
$\lambda_3>0$ and $\lambda_4>0$,
\begin{eqnarray}\label{852}
E_{31}&=& -\int_0^t\int_{0}^{1}
   [u(b_2y)/b_2^2-u(b_1y)/b_1^2]w(s,y)\theta_{y}^{b_1}(s,y)dyds\nonumber\\
   &\leq &\frac{L^2(b_2-b_1)^2}{4\lambda_3}\int_0^t\int_0^1
  [\theta_{y}^{b_1}(s,y)]^2dyds\nonumber\\
  && +\lambda_3\int_0^t\int_0^1 [w_y^2(s,y)+w^2(s,y)]dyds
\end{eqnarray}
and
\begin{eqnarray}\label{853}
 E_{33}&=& -\int_0^t\int_{0}^{m/b_2}
[u(b_2y)/b_2^2-u(b_1y)/b_1^2]'w(s,y)\theta_{y}^{b_1}(s,y)dyds\nonumber\\
&\leq &\frac{L^2(b_2-b_1)^2}{4\lambda_4}\int_0^t\int_0^1
[\theta_{y}^{b_1}(s,y)]^2dyds\nonumber\\
&& +\lambda_4\int_0^t\int_0^1 [w_y^2(s,y)+w^2(s,y)]dyds.
\end{eqnarray}
There exists a constant $D_4>0$ such that
 $|[u(by)/b^2]'-[v(by)/b]|\leq D_4$ and
 $\lambda_5=\inf\limits_{b_1\leq b\leq b_2}\{u(by)/b^2\}>0$.
Then by the boundary conditions we estimate $\int_0^t\int_0^1
[\theta_{y}^{b}(s,y)]^2dyds$ for $b\in [b_1,b_2]$ as follows,
\begin{eqnarray}\label{854}
  0&=&\int_0^t\int_0^1 \theta_{t}^{b}(s,y)\theta^{b}(s,y)\nonumber\\
  && -[u(by)/b^2]\theta_{yy}^{b}(s,y)\theta^{b}(s,y)
  -[v(by)/b]\theta_{y}^{b}(s,y)\theta^{b}(s,y)dyds\nonumber\\
  &=&\frac{1}{2}\int_0^1 [\theta^{b}(s,y)]^2dy
  +\int_0^t\int_0^1 [u(by)/b^2] [\theta_{y}^{b}(s,y)]^2dyds\nonumber\\
  &&+\int_0^t\int_0^1
  [u(by)/b^2]'\theta_{y}^{b}(s,y)\theta^{b}(s,y)dyds\nonumber\\
  &&-\int_0^t\int_0^1 [v(by)/b]\theta_{y}^{b}(s,y)\theta^{b}(s,y)dyds\nonumber\\
  &\geq & \lambda_5 \int_0^t\int_0^1 [\theta_{y}^{b}(s,y)]^2dyds
  -\frac{\lambda_5}{2}\int_0^t\int_0^1 [\theta_{y}^{b}(s,y)]^2dyds\nonumber\\
  &&-\frac{1}{2\lambda_5}\int_0^t\int_0^1 [\theta^{b}(s,y)]^2dyds\nonumber\\
  &\geq &\frac{\lambda_5}{2}\int_0^t\int_0^1
  [\theta_{y}^{b}(s,y)]^2dyds -\frac{D_4}{2\lambda_5}
\end{eqnarray}
from which we deduce that
\begin{eqnarray}\label{855}
  \int_0^t\int_0^1 [\theta_{y}^{b}(s,y)]^2dyds\leq
  \frac{D_4}{\lambda_5^2}.
\end{eqnarray}
Therefore we conclude that $\int_0^t\int_0^1
[\theta_{y}^{b}(s,y)]^2dyds$ is bounded. Noticing that $w(s,y)\leq
2$ and
\begin{eqnarray}\label{856}
\lim\limits_{b_2\rightarrow b_1}\{|E_{32}|+|E_{34}|\}=0,
\end{eqnarray}
as well as using the equalities (\ref{851})-(\ref{856}), there
exists a positive function $B_1^{b_1}(b_2)$ such that
\begin{eqnarray*}
\lim\limits_{b_2\rightarrow b_1}B_1^{b_1}(b_2)=0
\end{eqnarray*}
and for $ 0\leq t\leq T$
\begin{eqnarray}\label{857}
E_3&=&\int_0^t\int_{0}^{1}[u(b_2y)/b_2^2-u(b_1y)/b_1^2]w(s,y)\theta_{yy}^{b_1}(t,y)dyds
\nonumber\\
&\leq & B_1^{b_1}(b_2)+(\lambda_3 + \lambda_4 )\int_0^t\int_0^1
[w_y^2(s,y)+w^2(s,y)]dyds.\nonumber\\
\end{eqnarray}
By the same way as in estimating $E_3$ we also find  a positive
function $B_2^{b_1}(b_2)$ such that
\begin{eqnarray*}
\lim\limits_{b_2\rightarrow b_1}B_2^{b_1}(b_2)=0
\end{eqnarray*}
and for any $\lambda_6 >0$
\begin{eqnarray}\label{858}
E_4&=&\int_0^t\int_0^1 [v(b_2y)/b_2-v(b_1y)/b_1]w(s,y)
\theta_y^{b_1}(s,y)dyds\nonumber\\
&\leq &\frac{L^2(b_2-b_1)^2}{4\lambda_6}\int_0^t\int_0^1
  [\theta_{y}^{b_1}(s,y)]^2dyds+\lambda_6\int_0^t\int_0^1 w^2(s,y)dyds\nonumber\\
&\leq &  B_2^{b_1}(b_2)+\lambda_6\int_0^t\int_0^1 w^2(s,y)dyds.
\end{eqnarray}
Choosing  $\lambda_1$, $\lambda_2$, $\lambda_3$ and $\lambda_4$
small enough such that $\lambda_1 D_3+ \lambda_2 + \lambda_3+
\lambda_4\leq D_2$, we deduce from (\ref{846})-(\ref{449}),
(\ref{857}) and (\ref{858}) that there exist positive constants
$C_1$ and $C_2$ such that
\begin{eqnarray*}
\int_0^1w^2(t,y)dy\leq
C_1\int_0^t\int_0^1w^2(s,y)dyds+C_2[B_1^{b_1}(b_2)+B_2^{b_1}(b_2)].
\end{eqnarray*}
Setting $F(t)=\int_0^t\int_0^1w^2(s,y)dyds$ and using the Gronwall
inequality,
\begin{eqnarray*}
F(t)\leq C_2[B_1^{b_1}(b_2)+B_2^{b_1}(b_2)]\exp\{C_1t\}.
\end{eqnarray*}
So
\begin{eqnarray*}
\lim\limits_{b_2\rightarrow
b_1}\int_0^t\int_0^1[\theta^{b_2}(s,y)-\theta^{b_1}(s,y)]^2dyds=0.
\end{eqnarray*}
Thus we complete the proof.\ \ $\Box$
 \vskip 10pt\noindent {\bf
Acknowledgements.} This work is supported by Project 10771114 of
NSFC, Project 20060003001 of SRFDP, the SRF for ROCS, SEM  and the
Korea Foundation for Advanced Studies. We would like to thank the
institutions for the generous financial support. Special thanks also
go to the participants of the seminar stochastic analysis and
finance at Tsinghua University for their feedbacks and useful
conversations.
 \setcounter{equation}{0}

\end{document}